\documentclass{custom-jocg-2023}

\usepackage{graphicx}
\usepackage[boxruled]{algorithm2e}
\usepackage[caption=false]{subfig}

\usepackage{amsmath}
\usepackage{amssymb}
\usepackage{amsfonts}
\usepackage{amsthm}

\usepackage{thm-restate}

\makeatletter
\def\cl@chapter{\@elt {chapter}} 
\makeatother

\usepackage{ifthen}
\newboolean{enable-backrefs} 
\setboolean{enable-backrefs}{true} 
\newcommand{\backrefnotcitedstring}{\relax}
\newcommand{\backrefcitedsinglestring}[1]{(Cited on page~#1.)}
\newcommand{\backrefcitedmultistring}[1]{(Cited on pages~#1.)}
\ifthenelse{\boolean{enable-backrefs}}%
{%
		\PassOptionsToPackage{hyperpageref}{backref}
		\usepackage{backref} 
		   \renewcommand*{\backref}[1]{}  
		   \renewcommand*{\backrefalt}[4]{
		      \ifcase #1 %
		         \backrefnotcitedstring%
		      \or%
		         \backrefcitedsinglestring{#2}%
		      \else%
		         \backrefcitedmultistring{#2}%
		      \fi}%
}{\relax}

\usepackage[capitalise]{cleveref}
\usepackage[misc]{ifsym}

\newcommand{\etal}{et~al.\xspace}

\theoremstyle{definition}
\newtheorem{observation}{Observation}
\theoremstyle{definition}
\newtheorem{myRest}{Restriction}
\crefname{myRest}{Restriction}{Restrictions}

\theoremstyle{plain}
\newtheorem{lemma}{Lemma}
\theoremstyle{definition}
\newtheorem{definition}{Definition}
\theoremstyle{plain}
\newtheorem{corollary}{Corollary}
\theoremstyle{plain}
\newtheorem{theorem}{Theorem}

\title{%
  \MakeUppercase{On the Restricted \texorpdfstring{$k$}{k}-Steiner Tree Problem}%
  \thanks{This research was funded in part by the Natural Sciences and Engineering Research Council of Canada (NSERC). A preliminary version of this paper appeared in COCOON $2020$ under the title ``On  the  restricted  $1$-Steiner  tree  problem''  \cite{DBLP:conf/cocoon/BoseDD20} (available online at \url{https://doi.org/10.1007/978-3-030-58150-3_36}), and the full version (published by Springer Nature in $2021$, issue date $2022$) in the Journal of Combinatorial Optimization (JOCO) under the title ``On the Restricted \texorpdfstring{$k$}{k}-Steiner Tree Problem'' \cite{DBLP:journals/jco/BoseDD22} (available online at \url{https://doi.org/10.1007/s10878-021-00808-z}). A full version of this paper also appeared as a chapter in Anthony D'Angelo's thesis \cite{anthonydangelo2023thesis}. This paper has undergone minor corrections since publication.}
}

\author{%
  Prosenjit~Bose,%
  \thanks{\affil{School of Computer Science, Carleton University, Ottawa, Canada}, 
          \email{jit@scs.carleton.ca},
          \url{http://www.scs.carleton.ca/\~jit}}\,
  Anthony~D'Angelo,%
  \thanks{\affil{Carleton University, Ottawa, Canada},
          \email{anthony.dangelo@carleton.ca}}\,
  Stephane~Durocher%
  \thanks{\affil{University of Manitoba, Winnipeg, Canada}, 
          \email{durocher@cs.umanitoba.ca}}
}

\begin{document}

\maketitle

\begin{abstract}
Given a set $P$ of $n$ points in $\mathbb{R}^2$ and an 
input line $\gamma$ in $\mathbb{R}^2$,
we present an algorithm that runs in 
optimal $\Theta(n\log n)$ time and 
$\Theta(n)$ space
to solve a restricted version of the $1$-Steiner tree problem.
Our algorithm returns a minimum-weight tree interconnecting $P$
using at most one Steiner point $s \in \gamma$,
where edges are weighted by the Euclidean distance between their
endpoints.
We then extend the result to $j$ input lines.
Following this, we show how the algorithm of 
Brazil \etal\ \cite{brazil2015generalised} that solves
the $k$-Steiner tree problem in $\mathbb{R}^2$
in $O(n^{2k})$ time can be 
adapted to our setting.
For $k>1$, restricting the (at most) $k$ Steiner points to 
lie on an input line, the runtime becomes $O(n^{k})$.
Next we show how the results of 
Brazil \etal\ \cite{brazil2015generalised}
allow us to maintain the same time and space bounds
while extending to some non-Euclidean norms and 
different tree cost functions.
Lastly, we extend the result to $j$ input \emph{curves}.
\end{abstract}

\section{Introduction}
\label{sec: intro}
Finding the shortest interconnecting network for a given set 
of points is an interesting optimization problem 
with various applications for anyone seeking connectivity while
at the same time
concerned with conserving resources.
Sometimes we are able to introduce new points into the point set
such that the sum of the lengths of the edges in the
interconnecting network is reduced.
These extra points are called \emph{\textbf{Steiner points}}.
However, selecting optimal positions for 
these Steiner points and how many to place is 
NP-hard \cite{DBLP:journals/jco/BrazilTW00,brazilSteinerSurveyBook,garey1977complexity,DBLP:journals/siamdm/RubinsteinTW97}, 
and so a natural
question is: \emph{What is the shortest spanning 
network that can be constructed 
by adding only $k$ Steiner points to the given set of points?}
This is the \emph{\textbf{``$k$-Steiner tree" problem}}.

Consider a set $P$ of $n$ points in $\mathbb{R}^2$, 
which are also
called \emph{\textbf{terminals}} in the Steiner tree literature.
The \textbf{Minimum Spanning Tree} (\textbf{MST}) 
problem is to find the minimum-weight tree interconnecting $P$,
where edges are weighted by the Euclidean distance between their
endpoints.
Let $\operatorname{MST}(P)$ be a Euclidean
minimum spanning tree on $P$ and let $|\operatorname{MST}(P)|$ 
be the sum of its edge-weights (also called the \emph{\textbf{length}} of the tree).
Imagine we are given another set $S$ of points in $\mathbb{R}^2$.
$S$ is the set of Steiner points that we may use 
as intermediate nodes in 
addition to the points of $P$ to compute a minimum-weight 
interconnection of $P$.
An MST on the union of the terminals $P$ 
with some subset of Steiner points $S' \subseteq S$, 
i.e., $\operatorname{MST}(P \cup S')$,
is a \emph{\textbf{Steiner tree}}.
In the Euclidean \textbf{Minimum Steiner Tree} (\textbf{MStT}) problem, 
the goal is to find a subset $S' \subseteq S$ 
such that $|\operatorname{MST}(P \cup S')|$ is at most  
$|\operatorname{MST}(P \cup X)|$ for any $X \subseteq S$.
Such a minimum-weight tree is an MStT.
For our \emph{restricted} $k$-Steiner tree problem, we are given an input 
line $\gamma$ in $\mathbb{R}^2$;
the line $\gamma = S$ and the cardinality of $S'$ is at most $k$.

As $3$-D printing enters the mainstream, material-saving and time-saving 
printing algorithms are becoming more relevant.
Drawing on the study of MStTs,
Vanek et al.\ \cite{vanek2014clever} 
presented a geometric heuristic to create support-trees 
for $3$-D-printed objects where the forking points in these trees are 
solutions to a constrained Steiner tree problem.
Inspired by the work of Vanek \etal~as well as solutions
for the $1$-Steiner and $k$-Steiner tree 
problems in the $2$-D Euclidean plane
\cite{brazil2015generalised,brazilSteinerSurveyBook,oneStTreeProb},
we present an efficient algorithm to compute a solution 
for the $1$-Steiner tree problem where the placement
of the Steiner point is constrained to lie on an input line.
We present another motivating example.
Imagine we have a set $V$ of wireless nodes that 
must communicate by radio transmission. 
To transmit a longer distance to reach more distant nodes requires 
transmitting at a higher power. 
The MST of $V$ can be used to model a connected network 
that spans the nodes of $V$ while minimizing total power consumption. 
Suppose that an additional wireless node is available to be added to $V$, 
but that the new node's position is restricted to lie on a road $\gamma$ on 
which it will be delivered on a vehicle. 
Where on $\gamma$ should the additional node be positioned to minimize the 
total transmission power of the new network? 

We refer to our problem as a \emph{\textbf{$1$-Steiner tree problem restricted to
a line}}.
For our purposes, let an \textbf{\emph{optimal Steiner point}} be a point 
$s \in S$ such that 
$|\operatorname{MST}(P \cup \{s\})| \leq |\operatorname{MST}(P \cup \{u\})|$ 
for all $u \in S$.

\begin{description}
	\item[\texorpdfstring{$1$}{1}-Steiner Tree Problem 
	Restricted to a Line] 
	~\\	 
    Given a set $P$ of $n$ points in $\mathbb{R}^2$ 
	and a line $\gamma$ in $\mathbb{R}^2$, 
	select a point $s \in \gamma \cup \emptyset$
	that minimizes $|\operatorname{MST}(P \cup \{s\})|$.
\end{description}

A restricted version of our problem has been studied 
for the case when the input point set $P$ lies to one side of the given
input line and a point from the line \emph{must} be chosen.
Chen and Zhang gave an $O(n^2)$-time algorithm to solve this
problem \cite{CHEN2000867}.
Similar problems have also been studied by 
Li \etal~\cite{DBLP:journals/jco/LiLLWZ20} building on the research of
Holby~\cite{holby2017variations}. 
The two settings they study are: (a) 
the points of $P$ lie anywhere in $\mathbb{R}^2$ and \emph{must}
connect to the input line using any number of Steiner 
points, and any part of the
input line used in a spanning tree does not count towards its length;
and (b) the same problem, 
but the optimal line to minimize the network length 
is not given and must be computed.
Li \etal~provide 1.214-approximation\footnote{This means the length of their tree is at most 1.214 times the length of the optimal solution. Here they take advantage of the result of Chung and Graham \cite{doi:10.1111/j.1749-6632.1985.tb14564.x} showing that the MST is a 1.214-approximation (to three decimals) of the MStT.} 
algorithms for both (a) and (b)
in $O(n \log n)$ and $O(n^3 \log n)$ time respectively.
Our problem differs from the problems
of Chen and Zhang \cite{CHEN2000867}, 
Li \etal~\cite{DBLP:journals/jco/LiLLWZ20}, 
and Holby~\cite{holby2017variations} in the following ways:
(a) we have no restriction on the 
placement of the points of $P$ with 
respect to the input line $\gamma$;
(b) our problem does not require connecting to $\gamma$;
and (c) travel in our network has the same 
cost \emph{\textbf{on}} $\gamma$ as \emph{\textbf{off}} of it.
For example, if the points of the set $P$ were 
close to the line but far
from each other, then the solution of
Li \etal~\cite{DBLP:journals/jco/LiLLWZ20} would connect the points
to the line and get a tree with much lower weight/length than even the MStT.
Such an example is shown in \cref{fig: mst length vs sum distances line}: 
the MStT of points $\{a,b,c,d\}$ is the same as its MST since all triples
form angles larger than $\frac{2\pi}{3}$ 
\cite{brazilSteinerSurveyBook,gilbert1968steiner}. 
In our setting, the MST is the best solution for this point set, whereas
in the setting of Holby~\cite{holby2017variations} 
and Li \etal~\cite{DBLP:journals/jco/LiLLWZ20}, the best solution connects 
each input point directly to $\gamma$ to form a spanning tree between the 
points using pieces of $\gamma$.
The length of the MST is significantly larger than the length of the other 
solution since in their setting, only the edges connecting the points to 
$\gamma$ contribute to the length of the spanning tree.

\begin{figure}
\includegraphics[page=1,scale=0.24]{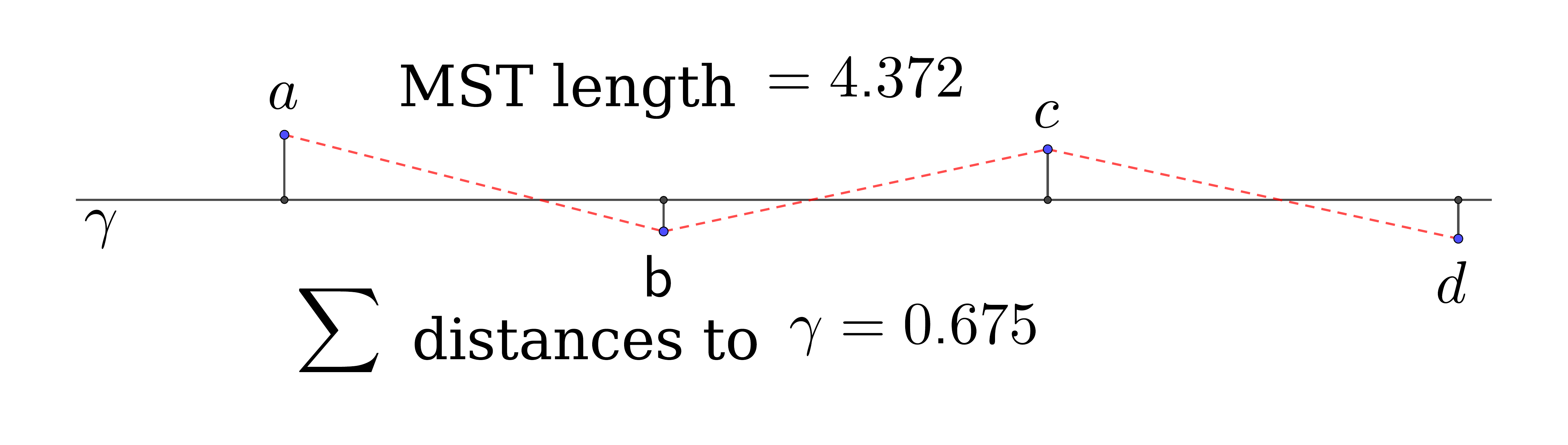}
\centering
\caption{Here we have $\gamma$ as the $x$-axis with
points $a=(0.489,0.237)$, 
$b=(1.865,-0.114)$, $c=(3.26,0.184)$, and $d=(4.75,-0.141)$.
Depicted are the MST of $\{a,b,c,d\}$ in red dashed line segments and its 
length; the input line $\gamma$; a spanning tree of $\{a,b,c,d\}$ 
connecting each point to $\gamma$ and the length of this spanning tree
for the setting of Holby~\cite{holby2017variations} 
and Li \etal~\cite{DBLP:journals/jco/LiLLWZ20}.
The numbers have been rounded to three decimal places.}
\label{fig: mst length vs sum distances line}
\end{figure}

In our algorithm we use a type of Voronoi diagram whose regions are
bounded by rays and segments.
We make a general position assumption that $\gamma$ is not collinear
with any ray or segment in the Voronoi diagrams.
In other words, the intersection of the rays and segments of 
these Voronoi diagrams with $\gamma$ is either empty or a single point.
We also assume that the edges of $\operatorname{MST}(P)$
have distinct weights.
In this paper we show the following theorem.

\begin{restatable}{theorem}{MainResultTheorem}
\label{theorem: main result}
Given a set $P$ of $n$ points in the Euclidean plane and a line $\gamma$,
\cref{alg: algo} computes in optimal $\Theta(n\log n)$ time and 
optimal $\Theta(n)$ space 
a minimum-weight tree connecting all points in $P$ using at most
one extra point $s \in \gamma$.
\end{restatable}

\cref{sec: steiner background results} reviews the tools 
and properties we will need for our algorithm.
\cref{sec: algorithm} presents \cref{alg: algo} and the proof
of \cref{theorem: main result}.
\cref{sec: generalization} 
shows how to adapt the existing $k$-Steiner tree algorithm of
Brazil \etal\ \cite{brazil2015generalised} to solve the problem
when the set of at most $k$ Steiner points is restricted to lie on $j$ input lines,
as well as
how the results of Brazil \etal\ \cite{brazil2015generalised}
allow us to generalize our algorithms to
norms other than Euclidean and
cost functions for the tree other than
minimizing the sum of edge-weights,
giving the following corollary.
The definitions of \emph{norm},
Restrictions \ref{rest: our four}, \ref{rest:one}, \ref{rest:two}, 
and \ref{rest:three},
and \emph{symmetric $\ell_1$-optimizable cost functions}
are found in \cref{sec: generalization}.
Informally, Restriction \ref{rest: our four} is that we can quickly 
solve a $k$-Steiner tree problem with a prespecified topology;
and 
Restrictions \ref{rest:one}, \ref{rest:two}, 
and \ref{rest:three} impose properties on the given norm, 
like the quick computations of intersections and bisectors, 
that make it behave somewhat like the familiar  Euclidean norm.

\begin{restatable*}{corollary}{generalCorollary}
\label{corollary: other costs}
Given:
\begin{itemize}
    \item a set $P$ of $n$ points in $\mathbb{R}^2$;
    \item a norm $\Vert \cdot \Vert$
that is compliant to Restrictions \ref{rest:one}, \ref{rest:two}, 
and \ref{rest:three};
    \item a symmetric $\ell_1$-optimizable cost function $\alpha$;
    \item $j$ lines $\Gamma = \{\gamma_1, \ldots \gamma_j\}$
\end{itemize}
By running \cref{alg: algo} for each $\gamma \in \Gamma$,
in $O(jn\log n)$ time and $O(n + j)$ space  
a minimum-weight tree with respect to $\alpha$ and 
$\Vert \cdot \Vert$ is computed that connects all points in $P$ 
using at most one extra point $s \in \bigcup_{i=1}^{j} \gamma_i$.

By running the algorithm of \cref{sec: k st points} for a constant integer $k>1$
under Restriction \ref{rest: our four},
the restricted $k$-Steiner tree problem is solved in 
$O((jn)^k)$ time and $O(jn)$ space with a Steiner 
set $S$ of at most $k$ points from  $\bigcup_{i=1}^{j} \gamma_i$.
\end{restatable*}

Lastly, in \cref{sec: from lines to curves} we show how to adapt 
the running time and space
of the algorithms when given $j$ input \emph{curves}
of a restricted class,
giving the following corollary.
The variables $\mu, t, g, c,$ and $q$
are defined in \cref{sec: from lines to curves}
and correspond to the complexities of certain 
primitive operations and the complexity of the
zone of the curves in an arrangement of lines.

\begin{restatable*}{corollary}{finalCorollary}
\label{corollary: k steiner points j curves}
Given:
\begin{itemize}
    \item a set $P$ of $n$ points in $\mathbb{R}^2$;
    \item a norm $\Vert \cdot \Vert$
that is compliant to Restrictions \ref{rest:one}, \ref{rest:two}, 
and \ref{rest:three};
    \item a symmetric $\ell_1$-optimizable cost function $\alpha$;
    \item $j$ input curves $\Gamma=\{\gamma_1, \ldots \gamma_j\}$
with maximum space complexity $\gamma_{sp}$
\end{itemize} 
By running \cref{alg: algo} for each $\gamma \in \Gamma$,
in $O(j(n\log n + \mu + t(g + c)))$ time 
and $O(n + j\gamma_{sp} + \mu_{sp} + t_{sp})$ 
space a minimum-weight tree with respect to $\alpha$ and $\Vert \cdot \Vert$
is computed
that connects all points in $P$ 
using at most one extra point $s \in \bigcup_{i=1}^{j} \gamma_i$. 

For a constant integer $k>1$, the algorithm of \cref{sec: k st points} solves
the restricted $k$-Steiner tree problem in
$O((j(g+c))^{k}q + j\mu + n\log n)$ time 
and $O(n + j\gamma_{sp} + j\mu_{sp} + q_{sp})$ space with a Steiner 
set $S$ of at most $k$ points from  $\bigcup_{i=1}^{j} \gamma_i$.
\end{restatable*}

In this paper we assume 
we can compute to any fixed precision in constant time and space 
the derivative and roots 
of a function that can be written using a constant number of
the following operators on one variable and the real numbers:
$+,-,*,/,$ and $\sqrt[c]{\cdot}$ for a constant $c$.

\section{Relevant Results}
\label{sec: steiner background results}

There has been a lot of research on Steiner trees in various dimensions,
metrics, norms, and under various constraints. 
See the surveys by Brazil \etal~\cite{brazil2014history} and 
Brazil and Zachariasen~\cite{brazilSteinerSurveyBook} for a good 
introduction.
In the general Euclidean case it has been shown that Steiner points
that reduce the length of the MST 
have degree three or four \cite{degree5stpts}.
There are results for building
Steiner trees when the \emph{terminal set} is restricted to 
\textit{zig-zags} \cite{booth1992steiner,du1983steiner},
\textit{curves} \cite{DBLP:journals/siamdm/RubinsteinTW97}, 
\textit{ladders} \cite{chung1978steiner},
and \textit{checkerboards} 
\cite{brazil1996minimal,brazil1997full,brazil1997minimal};
for when the angles between edges are constrained 
\cite{DBLP:journals/jco/BrazilTW00,brazilSteinerSurveyBook};
for obstacle-avoiding Steiner trees 
\cite{winter1993euclidean,winter1995euclidean,winterDIMACS,DBLP:journals/dam/WinterZN02,DBLP:conf/alenex/ZachariasenW99} (which include geodesic versions where the terminals, Steiner points, and tree are 
contained in polygons);
and for $k$-Steiner trees with $k$ as a fixed constant
where you can use at most $k$ Steiner points 
(for terminals and Steiner points in various normed planes including the  
$2$-D Euclidean plane, there is an $O(n^{2k})$-time algorithm)
\cite{brazil2015generalised,brazilSteinerSurveyBook,oneStTreeProb}.

\subsection{Tools}
\label{sec: tools}
Without loss of generality, we consider the positive $x$-axis to be the
basis for measuring angles, so that $0$ radians is the positive $x$-axis, 
$\frac{\pi}{3}$ radians is a counterclockwise rotation of 
the positive $x$-axis about the origin by $\frac{\pi}{3}$ radians,
etc.

\begin{observation}
\label{obs: msts}
Given a point set $V \subset \mathbb{R}^2$, 
if we build $\operatorname{MST}(V)$, each point $v \in V$ will have at
most six neighbours in the MST.
This is because, due to the sine law, 
for any two neighbours $w$ and $z$ of $v$ in $\operatorname{MST}(V)$
the angle $\angle wvz$ must be at least $\frac{\pi}{3}$ radians.
These potential neighbours can be found by dividing the plane up into
six interior-disjoint cones of angle $\frac{\pi}{3}$ all apexed on $v$.
The closest point of $V$ to $v$ in each cone is the potential neighbour 
of $v$ in the MST in that cone.
\end{observation}

Without loss of generality, suppose the input line $\gamma$
passes through the origin of the Euclidean plane with slope $0$.
This line can be parametrized by $x$-coordinates.
Let an \emph{\textbf{interval}} on $\gamma$ be the set of points on $\gamma$
in between and including two fixed $x$-coordinates, 
called the endpoints of the interval.
Our approach will be to divide the input line into $O(n)$ intervals
using a special kind of Voronoi diagram outlined below.
The intervals have the property that for any given interval $I$, 
if we compute $\operatorname{MST}(P \cup \{s\})$ for any $s \in I$,
the subset of possible neighbours of $s$ in the MST is 
constant.
For example, \cref{fig: interval neighbours} shows a set $V$ of input points
with the blue points labelled $p_i$ for $1\leq i \leq 6$, 
the input line $\gamma$, and a green interval $I$.
The plane is divided into six cones of $\frac{\pi}{3}$ radians, all apexed
on the red point $x \in I$.
In $\operatorname{MST}(V \cup \{x\})$,
if $x$ connects to a point in cone $i$, 
it connects to $p_i$.
The green interval $I$ has the property that this is true anywhere 
we slide $x$ and its cones in $I$.

\begin{figure}
\includegraphics[page=1,scale=0.16]{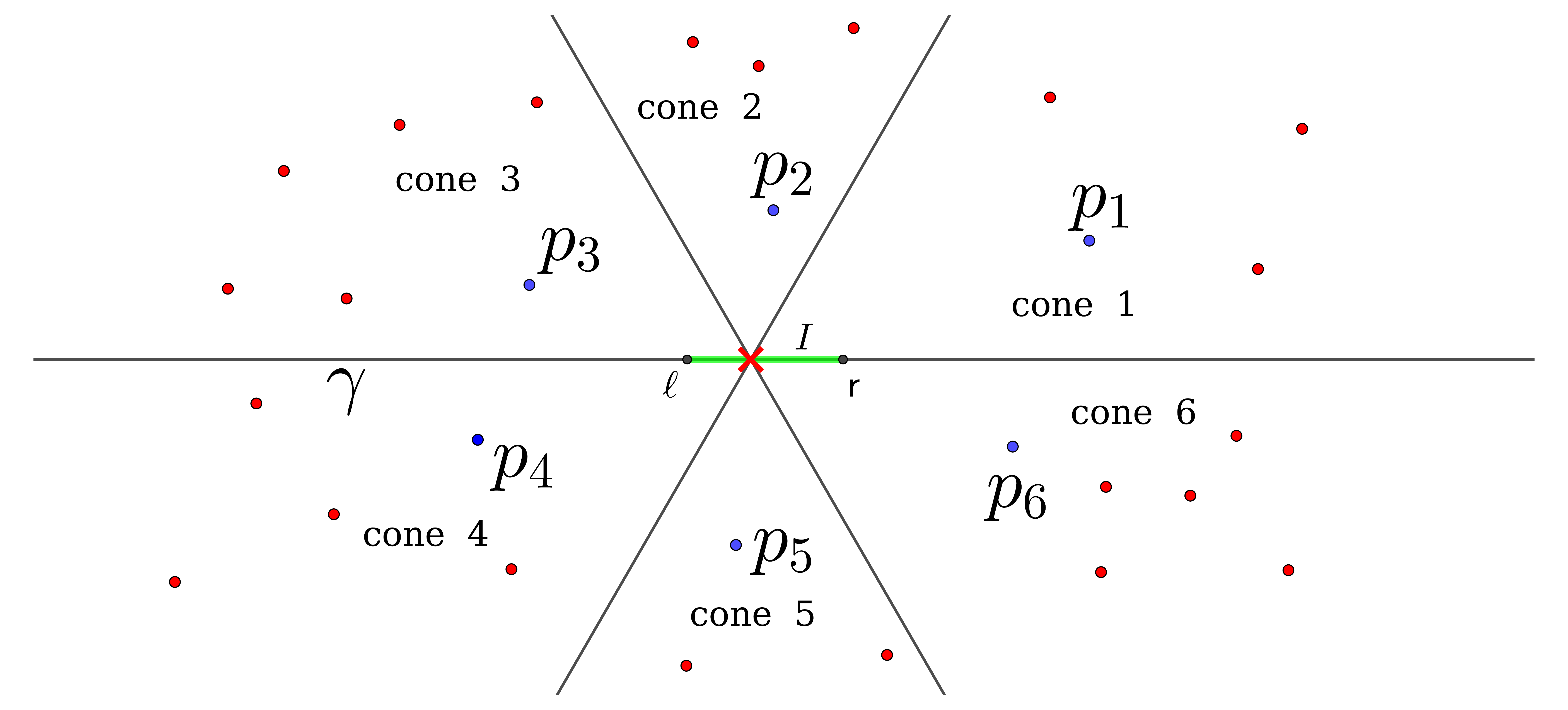}
\centering
\caption{Every point along the green interval $I$ of $\gamma$ 
(i.e., between the \emph{$\ell$} endpoint and the \emph{$r$} endpoint) 
has the same potential MST neighbour (the blue points) in the same cone.}
\label{fig: interval neighbours}
\end{figure}

\subsubsection{Oriented Voronoi Diagrams}
\label{subsec: OVDs}

The $1$-Steiner tree algorithm of Georgakopoulos and Papadimitriou
(we refer to this algorithm as \emph{GPA})
\cite{oneStTreeProb} works by 
subdividing the plane into regions defined by the cells of 
the \emph{Overlaid Oriented Voronoi Diagram} (overlaid \textbf{OVD}).\footnote{Georgakopoulos and Papadimitriou \cite{oneStTreeProb} refer to the overlaid OVD as \emph{Overlaid Oriented Dirichlet Cells}.} 
They show that the complexity of this diagram is $\Theta(n^2)$.
Refer to the cone $\mathbb{K}$ defining an OVD as an \emph{\textbf{OVD-cone}}.
Let $\mathbb{K}_{v}$ be a copy of the OVD-cone whose apex coincides
with point $v \in \mathbb{R}^2$.
OVDs are a type of Voronoi diagram made up of 
\emph{oriented Voronoi regions} 
(\textbf{OVRs}) where the OVR of a site $p \in P$ 
is the set of points $w \in \mathbb{R}^2$
for which $p$ is the closest site in $\mathbb{K}_w \cap P$
\cite{oneStTreeProb}.
If $\mathbb{K}_w \cap P = \varnothing$
we say $w$ belongs to an OVR whose site is the empty set.
These notions are illustrated in \cref{fig: OVD example}.

\begin{figure}
\includegraphics[page=1,scale=0.3]{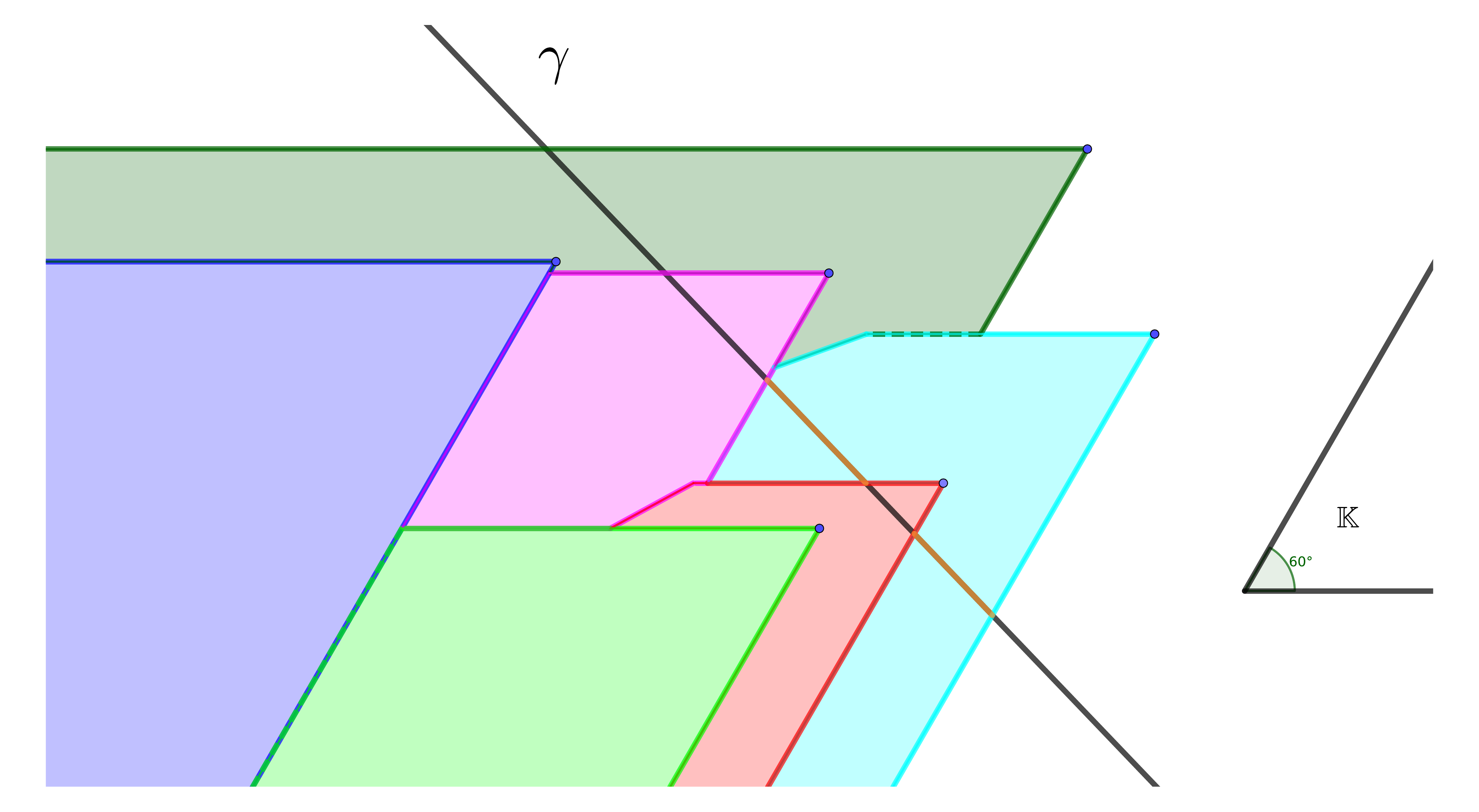}
\centering
\caption{An example of an OVD for six points defined by the OVD-cone 
$\mathbb{K}$ with bounding rays oriented towards $0$ and $\frac{\pi}{3}$.
The six sites (i.e., the points) are the blue top-right points 
of the coloured OVRs.
When intersected with $\gamma$, the OVD creates intervals along $\gamma$.
Each interval corresponds to exactly one OVR, but an OVR may create multiple 
intervals (for example, the light-blue OVR creates the two orange intervals).
The site corresponding to an interval outside of a coloured OVR is a 
special site represented by the empty set.}
\label{fig: OVD example}
\end{figure}

Chang \etal~\cite{DBLP:journals/ipl/ChangHT90}
show that the OVD for a given OVD-cone of angle $\frac{\pi}{3}$
(e.g., the OVD in \cref{fig: OVD example})
can be built in $O(n\log n)$ time using $O(n)$ space.
The OVD is comprised of segments and rays that are subsets of bisectors
and cone boundaries which bound the OVRs.
The size of the OVD is $O(n)$.

Since by 
\cref{obs: msts}
a vertex of the MST has a maximum degree of six, by overlaying
the six OVDs for the six cones of angle $\frac{\pi}{3}$
that subdivide the Euclidean plane (i.e., each of the six cones defines an 
orientation for a different OVD)
the GPA partitions the plane into $O(n^2)$ regions. 
Each of these regions has the property that if we place 
a Steiner point $s$ in the region, the points of $P$ associated with this 
region (up to six possible points) 
are the only possible neighbours 
of $s$ in the MStT (similar to the example in 
\cref{fig: interval neighbours}).
The GPA then iterates over each of these regions.
In region $R$, the GPA considers each subset of possible neighbours 
associated with $R$.
For each such subset it then computes the optimal location for a Steiner 
point whose neighbours are the elements of the subset, and then computes
the length of the MStT using that Steiner point, keeping track
of the best solution seen.
The generalized algorithm for placing $k$ Steiner points  
\cite{brazil2015generalised,brazilSteinerSurveyBook} 
essentially does the same thing $k$ times 
(by checking the topologies of the MStT
for all possible placements of $k$ points), 
but is more complicated (checking the effects that multiple Steiner points
have on the MStT is more complex).

\subsubsection{Updating Minimum Spanning Trees}
\label{sec: updating mst}

In order to avoid actually 
computing each of 
the candidate MSTs on the set of $P$ with the addition of our 
candidate Steiner points,
we instead compute the \emph{differences in length} 
between $\operatorname{MST}(P)$ and
the candidate MStTs.
Georgakopoulos and Papadimitriou \cite{oneStTreeProb}
similarly avoid repeated MST computations
by performing $O(n^2)$ preprocessing 
to allow them to answer queries of the following type in constant time: 
given that the edges $ab_1, ab_2, \ldots, ab_{j}$ are decreased by
$\delta_1, \delta_2, \ldots, \delta_j$ for constant $j$, 
what is the new MST? 
They then use these queries to find the length of the MStT
for each candidate Steiner point.
Refer to \cite{oneStTreeProb} for details.
Brazil et al.\ also perform some preprocessing in time 
$O(n^2)$ for $k=1$ and $O(n^3)$ otherwise
\cite{brazil2015generalised}.
However, using an approach involving
an auxiliary tree and lowest common ancestor (\textbf{LCA}) queries,
we can compute what we need in $o(n^2)$ time.
We first compute $\operatorname{MST}(P)$, build an auxiliary tree in 
$O(n\log n)$ time, and process the auxiliary tree in $O(n)$ 
time \cite{DBLP:journals/siamcomp/HarelT84}
to support LCA queries in $O(1)$ time
\cite{DBLP:journals/comgeo/BoseMNSZ04,monma1992transitions}.
This preprocessing helps us determine the benefit a given
Steiner point provides by computing in $O(1)$ time
the edges of the MST that disappear when the Steiner point
is added.
Details are outlined in the next section.

\section{Algorithm}
\label{sec: algorithm}
In this section we present \cref{alg: algo} and prove 
\cref{theorem: main result}.
\cref{alg: algo} computes OVDs for the six cones of angle $\frac{\pi}{3}$
that divide up the Euclidean plane (i.e., each of the six cones defines an 
orientation for a different OVD).
Though they can be overlaid in $O(n^2)$ time, 
we do not need to overlay them.
As mentioned in \cref{subsec: OVDs}, each OVD has $O(n)$ size and is 
therefore comprised of $O(n)$ rays and segments.
As illustrated in \cref{fig: OVD example}, 
intersecting any given OVD with a line $\gamma$ carves $\gamma$ up into
$O(n)$ intervals since we have $O(n)$ rays and segments, each of which 
intersects a line at most once.\footnote{This follows from the zone theorem
\cite{DBLP:journals/bit/ChazelleGL85,DBLP:books/lib/BergCKO08,DBLP:journals/siamcomp/EdelsbrunnerOS86,DBLP:journals/siamcomp/EdelsbrunnerSS93}.}
Each interval corresponds to an intersection of $\gamma$
with exactly one OVR of the OVD since OVDs are planar, 
but multiple non-adjacent intervals may be defined by the same OVR,
as in \cref{fig: OVD example}.
Therefore each interval $I$ is a subset of an OVR, 
and for every pair of points 
$u_1, u_2 \in I$ the closest point in 
$\mathbb{K}_{u_1} \cap P$ is the same as
in $\mathbb{K}_{u_2} \cap P$, where 
$\mathbb{K}$ is the OVD-cone of the OVD being considered.

If we do this with all six OVDs, $\gamma$ is subdivided into $O(n)$
intervals.
As in \cref{fig: interval neighbours},
each interval $I$ has the property that for any 
point $u \in I$, if we were to build $\operatorname{MST}(P \cup \{u\})$,
the ordered set of six potential neighbours is a 
constant set of constant size.\footnote{In other 
words, each $u \in I$ has the same constant-sized set of 
fixed candidate sub-topologies incident to $u$ that could be the 
result of $\operatorname{MST}(P \cup \{u\})$.}
Each element of this ordered set is defined by a different 
OVD and corresponds to
the closest point in $\mathbb{K}_u \cap P$.
In each interval we solve an optimization problem to find 
the optimal placement for a Steiner point in that interval
(i.e., minimize the sum of distances of potential neighbours
to the Steiner point) 
which takes $O(1)$ time since each of these $O(1)$
subproblems has $O(1)$ size.

\begin{figure}
	\subfloat[$\operatorname{MST}(\{a,b,c,d,e\})$ \label{fig: union mst mstt a}]{
		\includegraphics[width=0.31\textwidth]{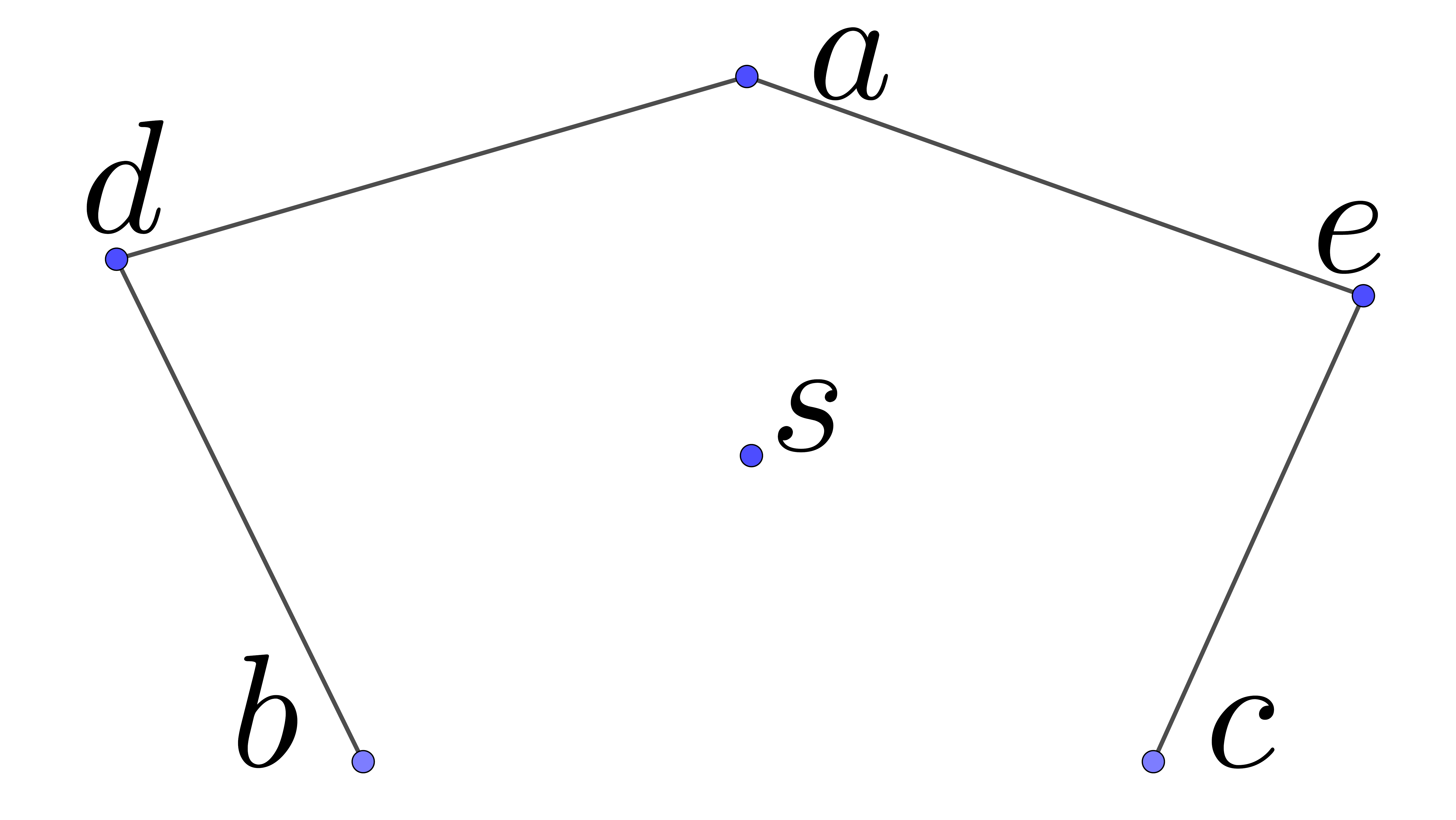}
		}
	\hfill
	\subfloat[$\operatorname{MST}(\{a,b,c,d,e,s\})$ \label{fig: union mst mstt b}]{
		\includegraphics[width=0.31\textwidth]{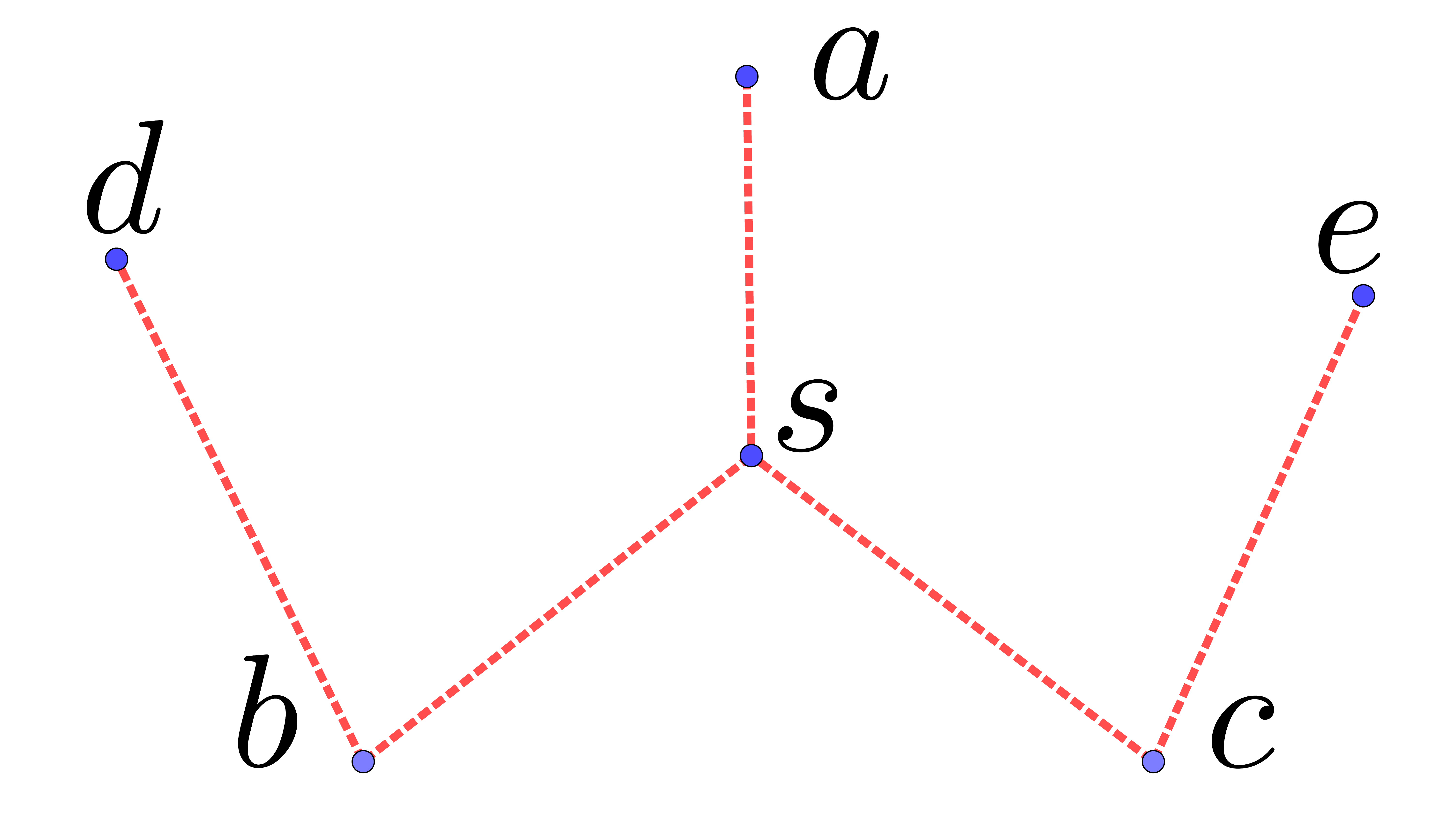}
		}
	\hfill
	\subfloat[Union of the trees from \cref{fig: union mst mstt a,fig: union mst mstt b} \label{fig: union mst mstt c}]{
		\includegraphics[width=0.31\textwidth]{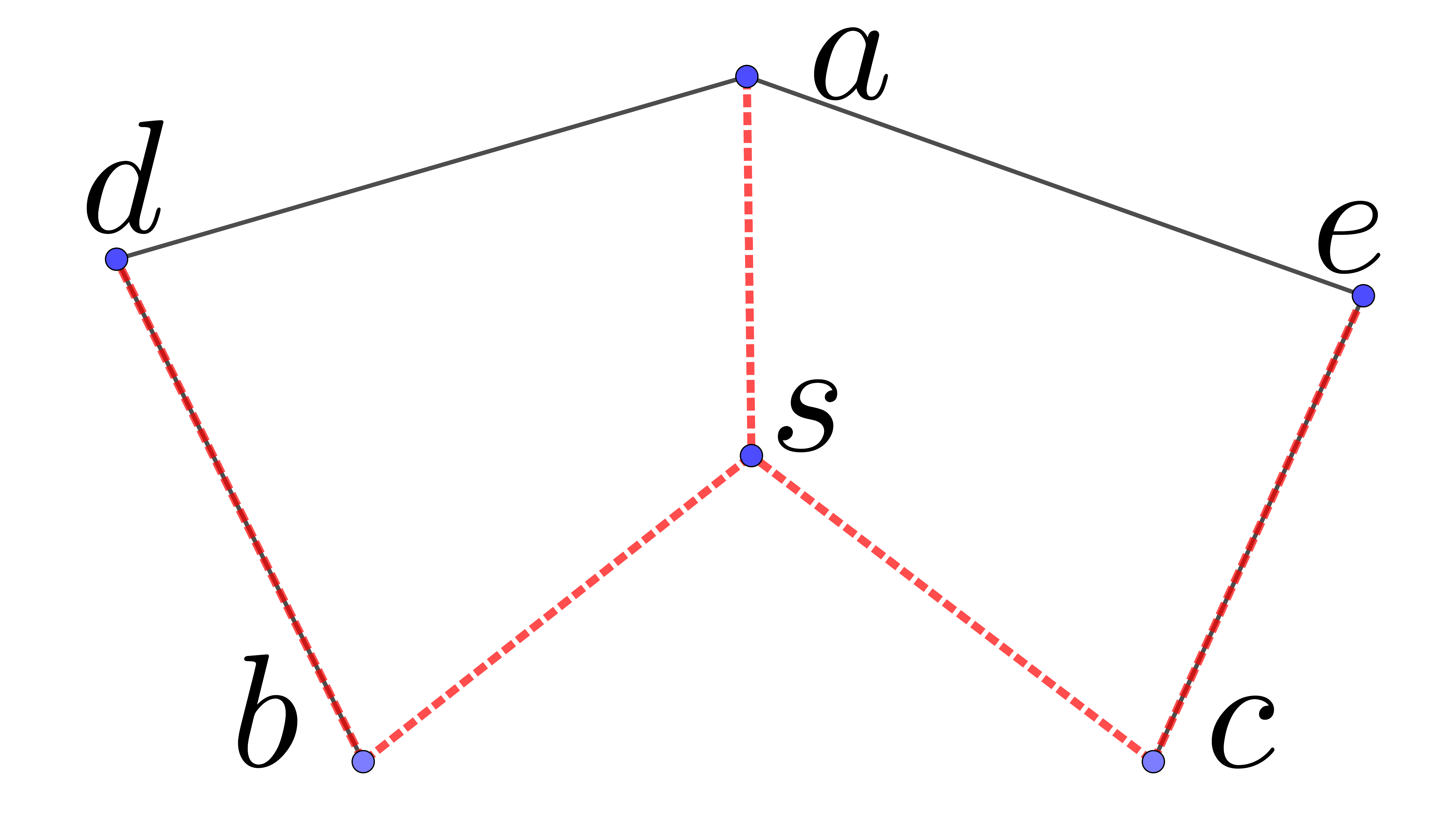}
		}
\caption{The union of the trees in 
\cref{fig: union mst mstt a,fig: union mst mstt b} gives the graph 
in \cref{fig: union mst mstt c}
with cycles $(s,b,d,a)$, $(s,a,e,c)$, and $(s,b,d,a,e,c)$ 
whose longest edges excluding $s$ are
$(d,a)$ and $(a,e)$.}
\label{fig: union mst mstt}
\end{figure}

Once we have computed an optimal placement for a Steiner point 
for each computed interval of our input line $\gamma$, 
we want to compute which one of these $O(n)$ candidates produces the MStT,
i.e., the candidate $s$ that produces the smallest length of the 
$\operatorname{MST}(P \cup \{s\})$.
Let $T^*$ be the union of $\operatorname{MST}(P)$ and 
$\operatorname{MST}(P \cup \{s\})$, as in \cref{fig: union mst mstt}.
For a candidate Steiner point $s$, the savings gained by including
$s$ in the MST are calculated by 
computing the sum $\Delta_s$ of the lengths of the longest edge on each cycle of $T^*$ 
(excluding the edges incident to $s$),
and then subtracting the sum $\sigma_s$ of the 
lengths of the edges incident to $s$ in
$\operatorname{MST}(P \cup \{s\})$.
For example, in \cref{fig: union mst mstt c}, the candidate edges on the
left cycle are $(b,d)$ and $(d,a)$, and on the right cycle they are 
$(a,e)$ and $(e,c)$;
we sum the lengths of the longest candidate edge from each cycle, 
i.e., $(d,a)$ and $(a,e)$,
and subtract the sum of the lengths of edges $(s,a)$, $(s,b)$, and $(s,c)$
to calculate the savings we get from choosing $s$ as the 
solution Steiner point.
Note that the longest edge on the cycle $(s,b,d,a,e,c)$ is either $(d,a)$ or $(a,e)$.
As will be seen in the proof of \cref{theorem: main result},
the sum $\sigma_s$ is computed 
when determining $s$.
What remains to find are the lengths of the longest edges
of $\operatorname{MST}(P)$ on the cycles of $T^*$ 
which will then be used to compute $\Delta_s$.

\clearpage

\begin{restatable}[{Bose \etal $2004$ \cite[paraphrased Theorem $2$]{DBLP:journals/comgeo/BoseMNSZ04}}]{lemma}{mstUpdate}
\label{lemma: mstUpdate}
A set of $n$ points in $\mathbb{R}^2$ can be preprocessed
in $O(n\log n)$ time into a data structure of size $O(n)$ such that
the longest edge on the path between any two points in the MST 
can be computed in $O(1)$ time.
\end{restatable}

\cref{lemma: mstUpdate} by Bose \etal \cite{DBLP:journals/comgeo/BoseMNSZ04}
tells us that with $O(n\log n)$ preprocessing of $\operatorname{MST}(P)$,
we can compute the sum $\Delta_s$ in $O(1)$ time for
each candidate Steiner point.\footnote{A similar result was shown in 
Monma and Suri \cite[Lemma $4.1$,~pg.~$277$]{monma1992transitions}.}
First an auxiliary binary tree is computed whose nodes correspond
to edge lengths and whose leaves correspond to points of $P$.
This tree has the property that the LCA of two leaves is the longest
edge on the path between them in $\operatorname{MST}(P)$.
Using $O(n)$ preprocessing time and space on the auxiliary tree
they perform $O(1)$-time 
LCA queries \cite{DBLP:conf/latin/BenderF00,DBLP:journals/siamcomp/HarelT84,DBLP:journals/siamcomp/SchieberV88}.

It was shown by Rubinstein \etal~\cite{degree5stpts} that the addition of
a Steiner point of degree five cannot appear in a least-cost
planar network, meaning we can restrict our
search to Steiner points of degree three or four in the Euclidean norm.
Let a simple cycle\footnote{A \emph{simple cycle} 
is a cycle where no vertex is repeated except the first vertex.} 
through points $s',s,s''$ be denoted 
$\mathcal{C}_{s'ss''}$.
In the next two lemmas, when referring to cycle 
$\mathcal{C}_{s'ss''}$, we are referring to the cycle in
$T^*$ that passes through the Steiner point $s$ and its two 
terminal-neighbours $s'$ and $s''$.
When we refer to the longest edge on the cycle
we mean the longest edge on the path between $s'$ and $s''$ 
in $\operatorname{MST}(P)$.

\begin{lemma}[Brazil \etal $2015$ \cite{brazil2015generalised}, rephrased Theorem $11$]
\label{lemma: deg 3 connected}
Let the Steiner point $s$ in $T^*$ have three neighbours
$a$, $b$, and $c$.
Using $O(1)$-time LCA queries,
removing the longest edges on the $\binom{3}{2}$ cycles
$\mathcal{C}_{bsa}$, $\mathcal{C}_{bsc}$, and
$\mathcal{C}_{csa}$ results in a tree which 
can be computed in $O(1)$ additional time and space.
\end{lemma}

\begin{lemma}[Brazil \etal $2015$ \cite{brazil2015generalised}, rephrased Theorem $11$]
\label{lemma: deg 4 connected}
Let the Steiner point $s$ in $T^*$ have four neighbours
$a$, $b$, $c$, and $d$.
Using $O(1)$-time LCA queries,
removing the longest edges on the $\binom{4}{2}$ cycles
$\mathcal{C}_{bsa}$, $\mathcal{C}_{csb}$, 
$\mathcal{C}_{dsc}$, $\mathcal{C}_{csa}$,
$\mathcal{C}_{dsb}$, and $\mathcal{C}_{dsa}$ results in a tree
which can be computed in $O(1)$ additional time and space.
\end{lemma}

By \cref{lemma: deg 3 connected,lemma: deg 4 connected} 
from Brazil \etal~\cite{brazil2015generalised},
removing the longest edges 
on the cycles between the neighbours of $s$ in $T^*$
results in a tree (i.e., one connected component).

We are now ready to finish proving \cref{theorem: main result}.
Though stated for input \emph{lines}, it is clear that the same result
applies for \emph{rays} and \emph{line segments}.
Rather than creating point-location data structures in the OVDs to locate
the endpoints and thus determine the correct labels (e.g., in the case
where $\gamma$ is a segment that lies completely in OVRs and has no 
intersections with the OVD boundaries), we notice that the intersections
with the OVDs give us endpoints for intervals on $\gamma$.
We can extend the line through $\gamma$ and run \cref{alg: algo} 
on this extension, adding special interval-delimiting points 
corresponding to the endpoints of the ray or line segment, and then only 
process intervals within the interval-delimiting points.
These special interval-delimiting points do not affect the labels of the 
intervals on either side of them.

\begin{algorithm}

	\SetKwInOut{Input}{input}\SetKwInOut{Output}{output}
	
	\Input{set $P \subset \mathbb{R}^2$ of $n$ points, a line $\gamma$}
	\Output{the MStT for $P$ using at most one Steiner point $s \in \gamma$}
	\BlankLine
	$s = (\infty, \infty)$\;
	$T = \operatorname{MST}(P)$\;
	$T' =$ longest-edge auxiliary tree built from $T$\;
	$\mathbb{L} = \emptyset$\;
	\For{$i=0, i \leq 5, ++i$}{
		Compute $\operatorname{OVD}_{i}$ and augment the
		edges with labels 
		corresponding to the sites of the two adjacent OVRs\;
		$\mathbb{L} = \mathbb{L} \cup $ 
		set of rays and segments defining the edges of 
		$\operatorname{OVD}_{i}$\;
	}
	\For{each $\ell \in \mathbb{L}$}{
		Compute the intersections of $\ell$ with $\gamma$ 
        adding the label of $\ell$ to the appropriate intersection points (and implicitly the intervals)\; 
	}
    Sort the intersection points along $\gamma$ to create 
    the list of labelled intervals produced by these intersections\; 
    
	\For{each interval $I$ along $\gamma$}{
		
		\For{each subset of the labels of $I$ of sizes $3$ and $4$}{
			$u$ = an optimal Steiner point along $\gamma$ 
			for the subset considered\;
			\tcc{This computation of an optimal point also gives us the sum $\sigma_u$}			
			\If{$|\operatorname{MST}(P \cup \{u\})| < 
				|\operatorname{MST}(P \cup \{s\})|$}
				{\label{if: better solution}
				\tcc{The test in this condition is performed using the sum $\sigma_u$ and 
				using $T'$ to compute the sum $\Delta_u$}				
				$s$ = $u$\;
			}			
		}
	}
	\If{$|\operatorname{MST}(P \cup \{s\})| < |T|$}{	
		\Return{$\operatorname{MST}(P \cup \{s\})$}\;
	}
	\Else{
		\Return{$T$}\;	
	}
	\caption{Restricted 1-Steiner}
	\label{alg: algo}
\end{algorithm}

\MainResultTheorem*

\begin{proof}
The tree $T = \operatorname{MST}(P)$ and its length
are computed in $O(n\log n)$ time and $O(n)$ space
by computing the Delaunay triangulation
in $O(n\log n)$ time and $O(n)$ space and then computing
the MST of the Delaunay triangulation
in $O(n)$ time and space \cite{DBLP:journals/siamcomp/CheritonT76,DBLP:books/sp/PreparataS85}.
By \cref{lemma: mstUpdate}, in $O(n\log n)$ time and 
$O(n)$ space we compute the longest-edge 
auxiliary tree $T'$ and preprocess it to answer LCA queries in $O(1)$ time.
Each of the six $\operatorname{OVD}$s is then computed
in $O(n\log n)$ time and $O(n)$ space~\cite{brazil2015generalised,DBLP:journals/ipl/ChangHT90,DBLP:conf/compgeom/ChewD85}.
In $O(n)$ time and space we extract $\mathbb{L}$, 
the set of rays and segments defining each OVR 
of each $\operatorname{OVD}$.
While computing the OVDs, 
in $O(n)$ time we add labels to the boundary rays and segments
describing which OVD-cone defined them and the two sites corresponding to 
the two OVRs they border.

Since $\gamma$ is a line, it intersects any element of
$\mathbb{L}$ at most once.
Therefore, computing the intersections of $\gamma$ with $\mathbb{L}$ takes 
$O(n)$ time and space.
Assume without loss of generality that $\gamma$ is the $x$-axis.
Given our $O(n)$ intersection points, 
we can make a list of the $O(n)$
intervals they create along $\gamma$ in $O(n\log n)$ time and $O(n)$ space
by sorting the intersection points by $x$-coordinate and then walking
along $\gamma$. 
During this process we also use the labels of the elements of $\mathbb{L}$
to label each interval with its six potential neighbours described above
in $O(1)$ time per interval.

An optimal Steiner point has degree greater than two
by the triangle inequality.
Rubinstein et al.\ \cite{degree5stpts} showed
that an optimal Steiner point has degree no more than four.
Therefore an optimal Steiner point has degree three or four.
Once we have computed our labelled intervals,
we loop over each interval looking for the solution by finding the 
optimal placement of a Steiner point in the interval for a constant number of fixed
topologies.
Consider an interval $I$ and its set of 
potential neighbours $P' \subset P$ of
size at most six.
For each subset $\mathbb{P}$ of $P'$ of size three and four (of which
there are $O(1)$ such subsets),
we compute a constant number of  candidate optimal Steiner points in $\gamma$.
Note that $\gamma$ is actually a polynomial function, $\gamma (x)$.
Our computation is done using the following 
function 
$d_{\mathbb{P}}(x)$,
where $a_{x}$ and $a_{y}$
are the $x$- and $y$-coordinates of point $a$ respectively,
and $\gamma (x)$ is the evaluation of $\gamma$ at $x$:
$d_{\mathbb{P}}(x) = \sum_{a \in \mathbb{P}}\sqrt{(a_x - x)^2 + (a_y - \gamma (x))^2}$.
We then take the derivative of this function and solve 
for the global minima by finding the roots
within the domain specified by the endpoints of $I$.
Since the size of $\mathbb{P}$ is bounded by a constant and since
the degree of the polynomial $\gamma$ is a constant, 
this computation takes $O(1)$
time and $O(1)$ space and the number of global minima is a constant.
Note that the value of this function at a particular 
$x$ for a particular
$\mathbb{P}$ tells us the sum $\sigma_u$ 
of edge lengths from the point 
$u=(x, \gamma (x))$ to the points in $\mathbb{P}$.
We associate this value with $u$.
Out of the $O(1)$ candidate points, we choose the one for which
$d_{\mathbb{P}}(x)$ is minimum.
We can break ties arbitrarily, since a tie means the points offer
the same amount of savings to the MST since they both
have the same topology in the MST 
(meaning they have the same cycles in
$\operatorname{MST}(P) \cup \operatorname{MST}(P \cup \{u\})$), 
and since the value of $d_{\mathbb{P}}(x)$ being the same means that
the sum of adjacent edges is the same.

Once we have our $O(1)$ candidate optimal Steiner points for $I$,
we need to compare each one against our current best solution $s$.
In other words, for each candidate $u$
we need to compare $|\operatorname{MST}(P \cup \{u\})|$
with $|\operatorname{MST}(P \cup \{s\})|$.
We take advantage of the following:
if we compute the union of $\operatorname{MST}(P)$ and 
$\operatorname{MST}(P \cup \{u\})$
we get at most $\binom{4}{2} = 6$ simple cycles 
through $u$.
Let this connected set of cycles be $Q$.
We have 
$|\operatorname{MST}(P \cup \{u\})| = |\operatorname{MST}(P)| +
\sigma_u
- \Delta_u$,
where $\Delta_u$ is the sum of the longest edge in each cycle of $Q$ 
excluding from consideration the edges incident to $u$.
By \cref{lemma: mstUpdate}, we can compute $\Delta_u$ in $O(1)$ time using $T'$.
By \cref{lemma: deg 3 connected,lemma: deg 4 connected},
removing the longest edge from each cycle of $Q$ results in a tree.
If $|\operatorname{MST}(P \cup \{u\})| < 
|\operatorname{MST}(P \cup \{s\})|$ we set $s=u$.

Finally, we check if 
$|\operatorname{MST}(P \cup \{s\})| < |T|$.
If so, we return $\operatorname{MST}(P \cup \{s\})$.
Otherwise we return $T$.

Now we show the space and time optimality.
The $\Omega(n)$-space lower bound comes from the fact that we have to
read in the input.
The $\Omega(n\log n)$-time lower bound comes from a reduction from
the \emph{closest pair} problem (\textbf{CPP}).
The CPP is where we are given $n$ points in $\mathbb{R}^2$ and we are
supposed to return a closest pair with respect to Euclidean distance.
The CPP has an $\Omega(n \log n)$-time 
lower bound \cite[Theorem $5.2$]{DBLP:books/sp/PreparataS85}.
Indeed, given an instance of CPP, we can transform it into our problem 
in $O(n)$ time by using the points as the input points $P$ 
and choosing an arbitrary $\gamma$.

Given the solution to our problem, we can find a closest pair
in $O(n)$ time by walking over the resulting tree.
First, remove the Steiner point (if any) and its incident edges
to break our tree up into $O(1)$ connected components.
Consider one of these components $\mathbb{C}$.
$\mathbb{C}$ may contain both points of multiple closest pairs, 
or none.
Imagine $\mathbb{C}$ contained both points for exactly one closest pair.
Then the edge connecting them will be in $\mathbb{C}$ and
it will be the edge with minimum-weight in $\mathbb{C}$;
otherwise it contradicts that we had a minimum-weight tree.
Imagine $\mathbb{C}$ contained both points for multiple closest pairs.
Pick one of the closest pairs.
If $\mathbb{C}$ does not contain the edge $e$
connecting the two points of the 
pair, then there is a path between them in $\mathbb{C}$ consisting 
of minimum-weight edges (whose weights match $e$) 
connecting other closest pairs;
otherwise we contradict the minimality of our tree
or that both points were in the same connected component.
If no component contains both points of a closest pair, then 
the path between a closest pair goes through the Steiner point.
Once again, choose a closest pair $(a,b)$ and
let the edge connecting this closest pair be $e$.
Due to the minimality of our tree, the weight of 
every edge on the path between $a$ and $b$ is no more than that of $e$.
However, since no component contains a closest pair, that means
that $a$ and $b$ are incident to the Steiner point.
Therefore, we get a solution to the CPP by walking over our resulting
tree and returning the minimum among a minimum-weight edge connecting 
neighbours of the Steiner point and a minimum-weight edge seen
walking through our tree excluding edges incident to the Steiner point.
\end{proof}

\begin{restatable}{corollary}{jLinesCorollary}
\label{corollary: j lines}
Given a set $P$ of $n$ points in the Euclidean plane
and $j$  lines $\Gamma = \{\gamma_{1}, \ldots \gamma_{j}\}$,
by running \cref{alg: algo} for each $\gamma \in \Gamma$, 
in $O(jn\log n)$ time and $O(n + j)$ space 
a minimum-weight tree is computed
that connects all points in $P$ 
using at most one point $s \in \bigcup_{i=1}^{j} \gamma_i$.
\end{restatable}

Until this point, we have implicitly assumed
that the $\operatorname{MST}(P)$ is unique. 
The algorithm still produces a correct result if 
this is not the case.
To see this, consider the case when 
there are multiple MSTs for $P$.
Imagine we have computed a Steiner tree $\tau$
using a MST $T$ of $P$ and let the Steiner point be $s$.
Consider the overlay $T^* = T \cup \tau$.
Now consider a Steiner tree $\tau'$ 
produced using a different MST $T'$ of $P$. 
Let the Steiner point of $\tau'$ also be $s$ and let the set of neighbours
of $s$ be the same in $\tau$ and $\tau'$.
Consider the overlay $T^\clubsuit = T' \cup \tau'$.
If we compare the differences between $T^*$ and $T^\clubsuit$, 
there are two cases.
In the first case, the cycles through $s$ are all the same, 
in which case the change in the MST did not affect 
the Steiner tree computation.
In the second case, at least one cycle through $s$ and its neighbours
is different.
This means that for some pair of neighbours of $s$, 
there is an edge in $T$ on the path between them that is swapped
for an edge of equal weight between some other pair 
of vertices of $P$ in $T'$.
The weight of the two trees $T$ and $T'$ is the same.
Since $s$ is connected to the same neighbours,
$\sigma_s$ is the same for both $\tau$ and $\tau'$.
So what may change is $\Delta_s$.
However, the values that compose $\Delta_s$ are selected using
the auxiliary tree that gives the bottleneck edge between two vertices
of the MST.
Since a MST on a point set $P$ is transformed into another
MST on $P$ by substituting one edge at a time 
for one of equal weight,
if the two neighbours of $s$ have bottleneck edges of different weights
in $T$ and $T'$, that implies that one of the MSTs is shorter than
the other.
Thus, $\Delta_s$ is the same for both $\tau$ and $\tau'$,
and $|\tau|=|\tau'|$.

There is another case to consider.
Let the Steiner points for $\tau$ and $\tau'$ be $s$ and 
$s'$ respectively.
They may or may not be the same point,
but let them have different sets of neighbours.
If $|\tau|\neq |\tau'|$, then
the savings for one tree are larger than the other; 
i.e., without loss of generality,
$\Delta_s - \sigma_s < \Delta_{s'} - \sigma_{s'}$.
However, when computing $\tau$, the algorithm considered $s'$
and its neighbour set and so computed the same $\sigma_{s'}$.
That implies that the sum of bottleneck edges is different between the
two trees $T$ and $T'$, which (as mentioned above) is a contradiction.
Thus, the algorithm produces a MStT even if 
$\operatorname{MST}(P)$ is not unique.

\section{Towards Generalization}
\label{sec: generalization}
In \cref{sec: k st points} we show how
to adapt the $O(n^{2k})$-time algorithm of 
Brazil \etal~\cite{brazil2015generalised}
to solve the $k$-Steiner tree problem for $k>1$ in $O((jn)^k)$ time
when the Steiner points
are restricted to lie on $j$ input lines
and when the input abides by the restrictions 
imposed by Brazil \etal
In \cref{sec: norms,sec: cost} we show that 
under certain restrictions, we can apply the results of
Brazil \etal~\cite{brazil2015generalised} to our restricted $k$-Steiner tree
problem allowing \cref{alg: algo} and the adapted $k$-Steiner tree algorithm
to maintain the same time and space bounds
while solving the problem for norms other than 
Euclidean and tree cost functions other than the sum of the edge-weights.
Finally, in \cref{sec: from lines to curves} we show that we
can solve the restricted $k$-Steiner tree problem when the set of Steiner
points is constrained to
a restricted class of $j$ input curves rather than lines.
The running time comes to depend on $j$, the 
runtime of certain primitive operations, and the complexity of the
zone of the curves in an arrangement of lines.

\subsection{\texorpdfstring{$k$}{k} Steiner Points}
\label{sec: k st points}
Given a set $P$ of $n$ points in the Euclidean plane $\mathbb{R}^2$
and a fixed constant positive integer $k$,
the algorithm of Brazil \etal \cite{brazil2015generalised} solves
the $k$-Steiner tree problem in $O(n^{2k})$ 
time with an $O(n^2)$-time preprocessing step for $k=1$ and an
$O(n^3)$-time preprocessing step for $k>1$.\footnote{As summarized in the survey by Brazil and Zachariasen
\cite{brazilSteinerSurveyBook}, for the bottleneck $k$-Steiner tree problem, 
i.e., the $k$-Steiner tree problem where the goal is to minimize the
length of the longest edge of the resultant tree for the given norm,
Bae \etal \cite{DBLP:journals/ipl/BaeLC10}
presented a $\Theta(n\log n)$-time and $O(n^2)$-time algorithm 
for $k=1$ and $k=2$ respectively in the Euclidean plane,
while Bae \etal \cite{DBLP:journals/algorithmica/BaeCLT11}
presented an $O((k^{5k}2^{O(k)})(n^k + n\log n))$-time algorithm for the $L_p$ metric
for a fixed rational $p$ with $1<p<\infty$, as well as
an $O(n\log n)$-time algorithm for $k=1$ in the $L_1$ metric, and  
an $O((7k+1)^{7k-1}k^4\cdot n\log^2 n)$-time algorithm for the $L_1$ and $L_{\infty}$ metrics.}
Below we show that
straightforward adjustments to the algorithm of 
Brazil \etal~\cite{brazil2015generalised}
allow their algorithm to be used when Steiner
points are constrained to lie on $j$ input lines 
$\gamma_1, \ldots \gamma_j$.

\begin{definition}[Fixed topology $k$-Steiner tree problem, Brazil \etal $2015$ \cite{brazil2015generalised} \textsection $4$]
\label{def: fixed top stein}
Given a set $A$ of at most $6k$ embedded terminals, a set $S$ of $k$
free (i.e., non-embedded) Steiner points, and a tree topology $\tau$ 
spanning $A \cup S$, find the coordinates of the Steiner points
(i.e., find the set $S$) such that the sum of the edge-weights in 
the tree is minimized.
\end{definition}

\begin{myRest}[Brazil \etal $2015$ \cite{brazil2015generalised}]
\label{rest: four}
A solution to the fixed topology $k$-Steiner tree problem is computable within
any fixed precision in $f(k)$ 
time, where $f(k)$ 
is a function dependent only on $k$.
\end{myRest}

Since $k$ is a constant, \cref{rest: four} says the 
fixed topology $k$-Steiner tree problem can be solved in $O(1)$ time.

\begin{definition}[Feasible internal topology, viable forest, minimum $F$-fixed spanning tree, 
Brazil \etal $2015$ \cite{brazil2015generalised}]
\label{def: viable forest f-fixed}
A forest $F$ is said to have 
a \emph{feasible internal topology} provided that
its node set is $A\cup S$ where $A \subseteq P$ 
is the set of leaves of $F$
and the Steiner points $S \subseteq \mathbb{R}^2$ are
the internal nodes.
A feasible internal topology with $|S| \leq k$ is called \emph{viable}
if and only if every Steiner point
in $S$ has at most six neighbours in $F$.
A shortest total-length tree $T_F$ on $P\cup S$ such that
the set of neighbours of Steiner points in $T_F$ is the same as in $F$
is referred to as a \emph{minimum $F$-fixed spanning tree}.
\end{definition}

Brazil \etal  \cite{brazil2015generalised} compute the overlaid OVD 
(as does the GPA \cite{oneStTreeProb})
which, similar to our intervals from \cref{alg: algo},
has the property that each region has associated with it a set of
points: one for each OVD overlaid.
For overlaid OVD region $R_i$, let this neighbour set be $C_P(R_i)$.
Under Restriction \ref{rest: four},
the algorithm of Brazil \etal~works as follows to produce a 
minimum-weight tree connecting all points in $P$ using at 
most $k$ points of $\mathbb{R}^2 \setminus P$:

\begin{enumerate}
\item \label{k st compute OVDs} Compute the overlaid OVD of $P$
\item \label{k st compute MST} Compute $T=\operatorname{MST}(P)$
\item \label{k st longest edge on path between points} Compute the longest edge on the path between every pair of points
$x$ and $y$ in $T$
\item \label{k st table if edge is on path between points} Compute a table $H$ whose entry is true for edge $e$ and terminals $y$ and $z$ if and only if $e$ is on the shortest path between $y$ and $z$ in $T$
\item \label{k st for each value of k and selection of ovd regions} For every $k'\leq k$ and each choice (with repetition) of $k'$ regions
$R_1,\ldots, R_{k'}$ of the overlaid OVD
\begin{enumerate}
	\item \label{k st associate steiner point with region} Associate the free Steiner point $s_i$ with region $R_i$
	\item \label{k st build subforest supergraph} Let $G$ be the graph consisting of the vertices 
	$\bigcup C_P(R_i) \cup \{s_1,\ldots,s_{k'}\}$, all edges 
	$(s_i,s_j), i\neq j,$ and all edges $(s_i, x)$ for every $x\in C_P(R_i)$
	\item \label{k st create set of viable subforests} Let $G^*$ be the set of all viable subforests of $G$
	\item \label{k st for each viable subforest} For each $\mathcal{F} \in G^*$
	\begin{enumerate}
		\item \label{k st solve fixed topology steiner problem} Solve the fixed topology $k$-Steiner tree problem for
		$\mathcal{F}$ to get the forest $F$
		\item \label{k st update MST info} Use $T$ and $F$ to compute a minimum $F$-fixed spanning 
		tree $T_F$ 
	\end{enumerate}
\end{enumerate}
\item \label{k st select solution tree} Let $T^*$ be a smallest total cost $T_F$ produced
\item \label{k st collect steiner points} Let $S$ be the set of Steiner points of $T^*$
\end{enumerate}

We now consider the worst-case time and space upper bounds of this algorithm.
Most of the time-bounds are discussed in 
Brazil et al.\ \cite{brazil2015generalised}.
\begin{itemize}
\item Step \ref{k st compute OVDs}, computing the overlaid OVD, 
runs in $O(n^2)$ time and space.\footnote{In the pseudocode for this algorithm in \cite{brazil2015generalised} there is a typo saying this is $O(n\log n)$ time, but the correct time bound, $O(n^2)$, is stated in other parts of the paper.}

\item Step \ref{k st compute MST}, 
computing $T=\operatorname{MST}(P)$, is done in
$O(n\log n)$ time and $O(n)$ space.

\item Step \ref{k st longest edge on path between points}, 
computing the longest
edge on the path between every pair of terminals in $T$, takes $O(n^2)$ time
and space.
This is done by computing the MST and doing a depth-first traversal 
from every terminal \cite[\textsection $1.4.3$ pg.\ $50$]{brazilSteinerSurveyBook}.

\item Step \ref{k st table if edge is on path between points}, computing the 
table $H$ that tells us if a specific edge is on the path 
in $T$ between two given points, takes $O(n^3)$ time and space. 
This is only done for $k>1$.

\item The loop on step 
\ref{k st for each value of k and selection of ovd regions}, 
choosing combinations for matching oriented OVD cells with up to $k$
Steiner points, has $O(n^{2k})$ iterations.
This follows because: 
\begin{itemize}
	\item each OVD is a linear arrangement of $O(n)$ complexity 
		\cite{DBLP:journals/ipl/ChangHT90};
	\item the overlay of a constant number of linear arrangements
			(each of whose complexity is $O(n)$)
			has $O(n^2)$ complexity 
			\cite[\textsection $8.3$]{DBLP:books/lib/BergCKO08} 
			\cite[\textsection $5.4$, \textsection $28$]{toth2017handbook};
	\item we are checking each topology that may be a solution.
	Since the solution will have $k^*$ Steiner points for $0 \leq k^* \leq k$, 
	we iterate through each topology for each choice of 
	up to $k$ Steiner points.
	This means associating one of the $O(n^2)$ faces of the overlaid OVD
	with each of our up to $k$ Steiner points.
	We allow multiple Steiner points to be associated with any given face 
	(i.e., a many-to-one relationship may exist)
	since Steiner points may be connected to other Steiner points. 
	Therefore, we have a sum of \emph{combinations with repetition} 
	of the up to $k$ Steiner points choosing from the $O(n^2)$ faces, giving us: \\
	\begin{align*}
	&\sum^{k}_{k' = 0}{\binom{n^2 + k' - 1}{k'}} \\ 
	=& \sum^{k}_{k' = 0}{\frac{(n^2+k'-1)!}{k'!(n^2-1)!}} \\
	=& \sum^{k}_{k' = 0}{\frac{n^2(n^2+1)(n^2+2)\cdots(n^2+k'-1)}{k'!}}
	\end{align*}
	By Stirling's approximation we have 
	$\left(\frac{k'}{e}\right)^{k'}\leq k'!$.
	Thus, we can upper-bound the expression as:
	\begin{align*}
	&\sum^{k}_{k' = 0}{\frac{n^2(n^2+1)(n^2+2)\cdots(n^2+k'-1)}{k'!}} \\
	\leq & \sum^{k}_{k' = 0}{\frac{e^{k'} n^2(n^2+1)(n^2+2)\cdots(n^2+k'-1)}{k'^{k'}}} \\
	\leq & \sum^{k}_{k' = 0}{\frac{3^{k'} (n^2+k'-1)^{k'}}{k'^{k'}}} \\	
	\leq & \sum^{k}_{k' = 0}{\left(\frac{3 (n^2+k')}{k'}\right)^{k'}} \\	
	\leq & \sum^{k}_{k' = 0}{\left(\frac{3n^2}{k'} + 3\right)^{k'}} \\	
	\leq & \sum^{k}_{k' = 0}{\left(\frac{4n^2}{k'}\right)^{k'}} \\	
	\leq & \sum^{k}_{k' = 0}{\left(\frac{4}{k'}\right)^{k'}n^{2k'}} \\
	\end{align*}	
	Since $k'$ is a constant, we have:
	\begin{align*}
	& \sum^{k}_{k' = 0}{\left(\frac{4}{k'}\right)^{k'}n^{2k'}} \\
	\in & \sum^{k}_{k' = 0}{O(n^{2k'})} \\
	\in & \quad O(kn^{2k})
	\end{align*}
	Since $k$ is a constant, this leaves us with $O(n^{2k})$.
\end{itemize}

\item Step \ref{k st associate steiner point with region}, 
associating the Steiner points with an oriented OVD cell, takes 
$O(k) = O(1)$ time and space.

\item Step \ref{k st build subforest supergraph},
building the complete graph on the Steiner points and adding the edges
to the appropriate candidate neighbours from their oriented OVD cells, 
takes $O(k^2) = O(1)$ time and space.

\item Step \ref{k st create set of viable subforests}, creating the set
of viable subforests, is mostly listed for exposition. 
Rather than computing this all at once, to save space
it would likely be set up as a function
call in the loop condition of step \ref{k st for each viable subforest} that 
returns the ``next'' viable subforest given some index counter, 
so we will not count the time and space requirements of this step.

\item Given $G$, grabbing the next viable subforest of $G^*$ in the loop
condition of step \ref{k st for each viable subforest} takes 
$O(k)$ time and space because a tree on $O(k)$ points has $O(k)$ edges and in the worst
case the next enumeration has to change all of them.
By Cayley's formula and the fact that each spanning forest is a 
subgraph of a spanning tree with $O(k)$ edges, the number of iterations of this loop
is $O(126^k\cdot k^{k}) = O(1)$ \cite{brazil2015generalised}.

\item By \cref{rest: four}, 
step \ref{k st solve fixed topology steiner problem} takes $f(k)=O(1)$
time and space.

\item The subroutine they invoke at step \ref{k st update MST info}
runs in $O(k^2)$ time and $O(k)$ space for $k=1$ and
$O(k^{2k +3}\cdot k!)$ time and $O(k^2)$ space
otherwise.
Either way, it is $O(1)$ time and space.

\item Step \ref{k st select solution tree} is also just for exposition.
At the end of step \ref{k st update MST info}, the cost of the new
tree is compared to the cost of the best tree seen which 
is then updated accordingly.
This takes $O(1)$ time and space in each iteration of the loop on step
\ref{k st for each viable subforest}.

\item Step \ref{k st collect steiner points}, reporting the Steiner points
in the solution tree, takes $O(k)=O(1)$ time and space.
\end{itemize}

We begin our adjustment.
Rather than computing the overlaid OVD, we compute $O(jn)$ intervals
on $\gamma_1, \ldots \gamma_j$ as in \cref{alg: algo}.
The main \emph{for-loop} at step 
\ref{k st for each value of k and selection of ovd regions} becomes: 
``for each choice (with repetition) of $k' \leq k$ \emph{intervals},
$I_1, \ldots I_{k'}$'', of which there are $O((jn)^k)$ iterations.\footnote{Note that since Steiner points may lie on different input lines,
we cannot simply run the basic algorithm $j$ times
as was done in \cref{corollary: j lines}.}
Steiner point $s_i$ is then associated with \emph{interval} $I_i$,
and rather than the ``neighbour set of \emph{region} $i$'', we use the
set of candidate neighbours of \emph{interval} $I_i$, much like we do
in \cref{alg: algo}.
The enumeration of the topologies in which we are interested
(i.e., the number of iterations of the inner-most \emph{for-loop}, step
\ref{k st for each viable subforest}) remains the same.
We replace their \cref{rest: four} with one reflecting our problem.
\begin{myRest}
\label{rest: our four}
A solution to the fixed topology $k$-Steiner tree problem where each Steiner
point is constrained to lie on its own specified line 
(not necessarily distinct from the lines of the other Steiner points)
is computable within
any fixed precision in $f(k)$ 
time and $f_{sp}(k)$ space, where $f(k)$
and $f_{sp}(k)$ are functions dependent only on $k$.
\end{myRest}
With \cref{rest: our four}, the first step of
the inner-most \emph{for-loop}, step
\ref{k st solve fixed topology steiner problem}, runs in $O(1)$ time and space. 
The second step of the inner-most \emph{for-loop},
step \ref{k st update MST info}, does not change and still
runs in $O(1)$ time and space.
However, as with our approach for \cref{alg: algo},
by \cref{lemma: mstUpdate}
we are able to use LCA queries
to avoid performing the $O(n^2)$ time and space
preprocessing in Step \ref{k st longest edge on path between points}.
Similarly, we are able to use LCA queries to avoid the
$O(n^3)$ time and space preprocessing
from Step \ref{k st table if edge is on path between points}.
In Brazil et al.\ \cite{brazil2015generalised},
for $k>1$, a TRUE/FALSE table is computed
in $O(n^3)$ time and space to be able to answer the following
query $\mathcal{Q}_{e,y,z}$ in $O(1)$ time:
given an edge $e$ of the MST $T$ and two vertices $y$ and $z$,
is $e$ on the path between $y$ and $z$ in $T$?

\clearpage

\begin{lemma}
\label{lemma: edge on path query}
A given set $P$ of $n$ points in $\mathbb{R}^2$
can be preprocessed in $O(n\log n)$ time and $O(n)$ space
to construct a data structure that supports
$\mathcal{Q}_{e,y,z}$ queries in $O(1)$ time.
\end{lemma}

\begin{proof}
As with \cref{theorem: main result},
we compute $T=\operatorname{MST}(P)$ in $O(n\log n)$ time and $O(n)$ space
and then we root $T$ at an arbitrary vertex.
We then preprocess $T$ in $O(n)$ 
time and space, like we did in \cref{sec: algorithm}, 
enabling us to perform $O(1)$-time LCA queries
\cite{DBLP:conf/latin/BenderF00,DBLP:journals/siamcomp/HarelT84,DBLP:journals/siamcomp/SchieberV88}.

In a tree, there is a unique path between two vertices.
LCA queries help us answer $\mathcal{Q}_{e,y,z}$ as follows.
We have two cases to consider. 
In the first case, one of the two query vertices is an ancestor
of the other in the rooted tree $T$; 
and in the other case, $\operatorname{LCA}(y,z) = r$ (i.e., the LCA of $y$ and $z$ is a third vertex $r$).
Let us consider the first case.
Without loss of generality, let $z$ be an ancestor of $y$ 
(in which case $\operatorname{LCA}(y,z) = z$).
To traverse $T$ from $y$ to $z$, we follow parent-pointers
until we reach $z$.
If $e=(a,b)$ is on the path between $y$ and $z$ in $T$, then it
would be seen when following the parent-pointers.
This means that,
in this case,
$e$ is on the path between $y$ and $z$
if and only if
$\operatorname{LCA}(y,z) = z$,
$\operatorname{LCA}(y,a) = a$, 
$\operatorname{LCA}(y,b) = b$,
$\operatorname{LCA}(a,z) = z$, and
$\operatorname{LCA}(b,z) = z$.

In the second case, $\operatorname{LCA}(y,z) = r$.
If $e=(a,b)$ is on the path between $y$ and $z$ in $T$, then it
would be seen either when following the parent-pointers from 
$y$ to $r$, or when following them from $z$ to $r$.
Without loss of generality, assume $e$ is seen when
traversing the parent-pointers from $z$ to $r$.
This means that,
in this case,
$e$ is on the path between $y$ and $z$
if and only if
$\operatorname{LCA}(y,z) = r$,
$\operatorname{LCA}(y,a) = r$, 
$\operatorname{LCA}(y,b) = r$,
$\operatorname{LCA}(a,z) = a$, and
$\operatorname{LCA}(b,z) = b$.
\end{proof}

The approaches of \cref{lemma: mstUpdate,lemma: edge on path query} 
using LCA queries apply both in our setting as well as
to the results 
of Brazil et al.\ \cite{brazil2015generalised},
reducing the time and space of steps
\ref{k st longest edge on path between points} and \ref{k st table if edge is on path between points}
from $O(n^2)$ for $k=1$ and $O(n^3)$ for $k>1$ to
$O(n\log n)$ time and $O(n)$ space for all $k$, thus also reducing
their overall space usage to the $O(n^2)$ space used
to build the overlaid OVD.
Making the appropriate substitutions in their analysis we get 
the following corollary for our scenario.

\begin{corollary}
\label{corollary: k steiner points}
Given a set $P$ of $n$ points in the Euclidean plane $\mathbb{R}^2$, 
a constant integer $k>1$,
and a set of $j$ input lines $\Gamma=\{\gamma_1, \ldots \gamma_j\}$,
under \cref{rest: our four} the modified MStT algorithm of 
Brazil \etal~\cite{brazil2015generalised}
solves the restricted $k$-Steiner tree problem in $O((jn)^k)$ time
and $O(jn)$ space
with a Steiner set $S$ of at most $k$ points from  $\bigcup_{i=1}^{j} \gamma_i$.
\end{corollary}

Brazil \etal~in \cite{brazil2015generalised}
actually study what they call the \emph{generalized $k$-Steiner tree} 
problem, presenting the algorithm outlined in this section
for the $k$-Steiner tree problem
allowing norms other than Euclidean and allowing different cost functions 
for the weight of a tree.\footnote{For the reader trying to get
a better understanding of the
material from Brazil \etal~\cite{brazil2015generalised}, 
the same material is presented a bit differently 
in the survey by Brazil and Zachariasen 
\cite{brazilSteinerSurveyBook}.}
Following the approach in \cite{brazil2015generalised}, our restricted
Steiner problem can also be solved in other norms and cost functions.
To achieve the runtime and space bounds of \cref{theorem: main result}, 
\cref{alg: algo} takes advantage of the fact that:
in the Euclidean norm, a vertex in $\operatorname{MST}(P)$ 
has a constant maximum degree; 
the plane can be partitioned the same way into a constant number of regions
around any point $u \in P$ and in each region the number of potential 
neighbours of $u$ in the MST is constant;
and the OVDs and MST can be computed in $O(n\log n)$ time and $O(n)$ space.

The results of Brazil \etal~\cite{brazil2015generalised} allow us to 
keep the same bounds as \cref{theorem: main result} and
\cref{corollary: k steiner points} with different norms
and cost functions.
The fact that the algorithm from \cref{corollary: k steiner points}
works for different norms and cost functions
(subject to the restrictions presented below) follows
from Brazil \etal~\cite{brazil2015generalised},
the fact that computing the intervals along the
lines in $\Gamma$ created by the OVDs does not depend on either the norm or tree
cost function, the fact that we assume that the weight of an edge can be calculated
in constant time, and the fact that finding the weight of
a solution tree can still be done in $O(n)$ time.
Below we show that the bounds on \cref{alg: algo}
do not change under the same set of norms and cost functions
considered by Brazil \etal

\subsection{Other Norms}
\label{sec: norms}
We begin with some definitions and notation 
from Brazil \etal \cite{brazil2015generalised}.
Let $\Vert \cdot \Vert$ be a given norm on $\mathbb{R}^2$, i.e., 
a function $\Vert \cdot \Vert : \mathbb{R}^2 \rightarrow \mathbb{R}$
that satisfies $\Vert x \Vert \geq 0$ for all $x \in \mathbb{R}^2$,
$\Vert x \Vert = 0$ if and only if $x=0$, 
$\Vert r x \Vert = |r| \cdot \Vert x \Vert$ for $r \in \mathbb{R}$ 
with $|\cdot|$ the absolute value function,
and $\Vert x + y \Vert \leq \Vert x \Vert + \Vert y \Vert$ for all
$x,y \in \mathbb{R}^2$.
The unit ball $B = \{x:\Vert x \Vert \leq 1\}$ is a centrally symmetric 
convex set.
Let $\partial(B)$ denote the boundary of $B$.

We now alter our definition of 
$\operatorname{MST}(\cdot)$ to use the given norm 
$\Vert \cdot \Vert$ rather than the Euclidean norm
when computing the length of the tree.
To ensure the MST and required OVDs 
can still be constructed in $O(n\log n)$ time and $O(n)$ space 
under the different norms, and that the maximum degree of a vertex in 
the MST is still a constant, Brazil \etal~impose some restrictions
\cite{brazil2015generalised}.
We also assume that computing the distance between two points in the given
norm takes constant time.

\begin{myRest}[Brazil \etal $2015$ \cite{brazil2015generalised}]
\label{rest:one}
The intersection points of any two translated copies of $\partial(B)$,
and the intersection points of any straight line and $\partial(B)$, can
be calculated to within any fixed precision in constant time.
\end{myRest}

\begin{lemma}[Brazil \etal $2015$ \cite{brazil2015generalised}, Lemma $3$]
\label{lemma: ball cut into six}
There exist six points $\{y_i : i = 0, . . . , 5\}$ 
on $\partial(B)$ such that for any pair
of rotationally consecutive ones, say $y_i , y_j$, 
we have $\Vert y_i - y_j \Vert = 1$.
Moreover, these six points are constructable.
\end{lemma}

\begin{definition}[Brazil \etal $2015$ \cite{brazil2015generalised}]
\label{def: cones}
For two directions $\phi_i$ and $\phi_j$ in the plane,
\textbf{\emph{$\mathbb{K}_y(\phi_i , \phi_j)$}} 
denotes the cone defined to be the set 
consisting of all rays emanating from $y$ in direction $\phi$, for
$\phi_i \leq \phi \leq \phi_j$.
For each $y_i$ from \cref{lemma: ball cut into six} let $\theta_i$
be the direction of the ray starting at the center of $B$ and going through
$y_i$.
We assume that the $\{\theta_i\}$ are ordered in a counterclockwise manner,
and two consecutive directions will be denoted by $\theta_i$ and 
$\theta_{i+1}$ (i.e., the mod $6$ notation will be omitted).
\end{definition}

\begin{lemma}[Brazil \etal $2015$ \cite{brazil2015generalised}, Lemma $5$]
\label{lemma: mst six neighbours}
Let $y$ be any point in the plane. 
Then there exists a $T = \operatorname{MST}(P \cup \{y\})$ with the 
following property: 
for each $i = 0, \ldots , 5$ there is at most one
point of $P$ that is adjacent to $y$ in $T$ and lies in cone 
$\mathbb{K}_y(\theta_i , \theta_{i+1})$, and this point is
a closest terminal to $y$ in the cone.
\end{lemma}

We conclude from \cref{lemma: mst six neighbours} that six OVDs
still suffice for \cref{alg: algo}.
The computation of the OVDs in \cref{alg: algo} must be modified
so that the defining cone for $\operatorname{OVD}_i$ is
$\mathbb{K}(\theta_i , \theta_{i+1})$ for $0 \leq i \leq 5$.

\begin{definition}[Chew and Drysdale $1985$ \cite{DBLP:conf/compgeom/ChewD85}, Brazil and Zachariasen $2015$ \cite{brazilSteinerSurveyBook}]
\label{def: abstract voronoi diagram}
Any closed convex curve $C$ bounding a region containing the origin $o$
defines a generalized Voronoi diagram \cite{DBLP:conf/compgeom/ChewD85}.\footnote{We can even allow $o$ to lie on $C$, but this means some points
are infinitely far from $o$ \cite[\textsection $4.3.1$ pg.\ $258$]{brazilSteinerSurveyBook}.}
This curve $C$ represents the boundary of a unit ball.
We then get a distance function $\delta_C$ (defined by $C$)
in the plane where the distance of $o$ to a point $v$ is $|ov|/|ov_o|$ where
$v_o$ is the point at which the ray starting at $o$ going through $v$
intersects $C$ and $|\cdot|$ is the Euclidean norm 
giving us the Euclidean distance between the two points.
\end{definition}

\begin{theorem}[Chew and Drysdale $1985$ \cite{DBLP:conf/compgeom/ChewD85}, Brazil and Zachariasen $2015$ \cite{brazilSteinerSurveyBook}]
\label{thm: abstract VDs}
The Voronoi diagram of $n$ points based on a closed convex shape
$C$ can be constructed in $O(n \log n)$ time and $O(n)$ 
space as long as the following operations can be performed in constant time:
\begin{enumerate}
\item Given two points $a$ and $b$, compute the bisector curve between them 
(i.e., the set of points equidistant from $a$ and $b$ using the $\delta_C$
distance function).
\item Given two such bisectors, compute their intersection(s).
\end{enumerate}
\end{theorem}

\begin{myRest}[Brazil \etal $2015$ \cite{brazil2015generalised}]
\label{rest:two}
Let $C'$ be any sector of $B$. 
Then, given any two points $a$ and $b$, we can compute the bisector curve 
between them (i.e., the set of points equidistant from $a$ and $b$ using the 
$\delta_{C'}$ distance function), and, given two such
bisectors, we can compute their intersection.
Moreover, these operations can be performed to within any fixed precision in 
constant time.
\end{myRest}

We make another general position assumption on the input points, 
namely that the bisector
of two points under the given norm does not contain regions.
Since \cref{alg: algo} assumes each OVD is an arrangement of lines,
we also require the next restriction.

\begin{myRest}[Brazil \etal $2015$ \cite{brazil2015generalised}]
\label{rest:three}
The shape of $B$ implies that the $i$\textsuperscript{th} OVD 
partition of any set of points is piecewise linear.
\end{myRest}

The first few steps of \cref{alg: algo} are:
\begin{enumerate}
\item compute $T=\operatorname{MST}(P)$
\item build the longest-edge auxiliary tree on $T$
\item compute the constant number of OVDs as labelled linear arrangements
\item compute the intervals of interest along $\gamma$
\end{enumerate} 

Norms that comply with \cref{rest:one,rest:two}
allow us to apply \cref{thm: abstract VDs} and compute
$T$ and the six OVDs in $O(n \log n)$ time and $O(n)$ space. 
The size of these OVDs is $O(n)$ 
\cite{DBLP:conf/compgeom/ChewD85}.
Once we have $T$, building the auxiliary tree needs no special care.
Norms that also comply with \cref{rest:three} produce OVDs
that are linear arrangements, so intersecting a line $\gamma$ 
with these piecewise-linear OVDs still creates $O(n)$ 
intervals and we can still compute them in
the time and space it takes to compute and sort the intersection points, i.e., $O(n\log n)$ time and $O(n)$ space.
A few examples of norms that satisfy \cref{rest:one,rest:two,rest:three} are
Euclidean, $L_1$, and $L_{\infty}$ \cite{brazil2015generalised}. 

The next step is computing candidate solution points in each interval
using our input norm and its distance function.
The time to compute optimal Steiner points for an interval
is still proportional to both: finding the roots of the derivative of 
the sum-of-distances function mentioned in \cref{alg: algo};
and computing the distance
between a terminal and a point on $\gamma$. 
For norms that respect our restrictions, this computation
takes $O(1)$ time and space when $\gamma$ is a line.  
As mentioned, for our compliant norms, a Steiner point will have at most
one neighbour in each OVD cone.
However, when computing the Steiner point for an interval
we must now consider neighbour subsets of sizes three to six  
points (of which there are still $O(1)$ such subsets).
The LCA queries are not affected by different norms.
Thus, we have the following corollary where the weight of a tree
is the sum of its edge-weights, which are equal to their length under
our given norm.

\begin{restatable}{corollary}{jNormCurvesCorollary}
\label{corollary: j norm curves}
Given:
\begin{itemize}
    \item a set $P$ of $n$ points in $\mathbb{R}^2$;
    \item a norm $\Vert \cdot \Vert$
that is compliant to Restrictions \ref{rest:one}, \ref{rest:two}, 
and \ref{rest:three};
    \item $j$ lines $\Gamma = \{\gamma_1, \ldots \gamma_j\}$
\end{itemize}
By running \cref{alg: algo} for each $\gamma \in \Gamma$,
in $O(jn\log n)$ time and $O(n + j)$ space 
a minimum-weight tree is computed that  connects all points in $P$ 
using at most one extra point $s \in \bigcup_{i=1}^{j} \gamma_i$.

By running the algorithm of \cref{sec: k st points} for a constant integer $k>1$
under Restriction \ref{rest: our four},
the restricted $k$-Steiner tree problem is solved in 
$O((jn)^k)$ time and $O(jn)$ space 
with a Steiner set $S$ of at most $k$ points from  $\bigcup_{i=1}^{j} \gamma_i$.
\end{restatable}

\subsection{Other Cost Functions}
\label{sec: cost}

Brazil \etal~\cite{brazil2015generalised}
also showed that we can use our results
for different tree cost functions.
We first review some definitions from 
Brazil \etal \cite{brazil2015generalised}.
\begin{definition}[Brazil \etal $2015$ \cite{brazil2015generalised}]
\label{def: edge length vector}
For a MStT $T$ with topology $\tau$ built on $P$ with Steiner set $S$, 
let $\textbf{e}_{\tau,P,S}$ be the vector whose components
are the edge-lengths of $T$ with respect to the given norm
$\Vert \cdot \Vert$.
Let $E(T)$ be the edge set of $T$ with cardinality $|E(T)|$.
\end{definition}
\begin{definition}[Brazil \etal $2015$ \cite{brazil2015generalised}]
\label{def: cost function}
Let the cost function
$\alpha: \mathbb{R}^{|E(T)|}_{+} \rightarrow \mathbb{R}$ be a symmetric
function (i.e., independent of the order of the components in the vector
on which it acts).
The cost of $T$ with respect to $\alpha$ is $\alpha(\textbf{e}_{\tau,P,S})$,
and $\min_{\tau, S}\alpha(\textbf{e}_{\tau,P,S})$ is the minimum cost
of any tree connecting $P$ using $|S|$ other points.
The power-$p$ cost function is $\alpha_p(\textbf{e}_{\tau,P,S}) = \sum_{i=1}^{|E(T)|} \Vert e_i \Vert ^ p $, where $e_i$
is the $i$\textsuperscript{th} element of $E(T)$,
and the cost of the longest edge is $\alpha_{\infty}(\textbf{e}_{\tau,P,S}) = \max_{i=1 \ldots |E(T)|} \Vert e_i \Vert$.
\end{definition}
\begin{definition}[Brazil \etal $2015$ \cite{brazil2015generalised}]
\label{def: l1 opt}
We say that a symmetric function $\alpha=\alpha(\textbf{e}_{\tau,P,S})$ 
is \emph{$\ell_1$-optimizable} if and only if there exist
$\tau^*$ and $S^*$ such that both
$\alpha(\textbf{e}_{\tau^*,P,S^*}) = \min_{\tau, S}\alpha(\textbf{e}_{\tau,P,S})$
and $\alpha_1(\textbf{e}_{\tau^*,P,S^*}) = \min_{\tau}\alpha_1(\textbf{e}_{\tau,P,S^*})$.
In other words, $\alpha$ is $\ell_1$-optimizable if and only if, for any given $P$,
there exists a Steiner set $S^*$ and tree topology interconnecting $P$ and $S^*$ 
that is both: minimum in cost with respect to $\alpha$; 
and is also an MST on $P \cup S^*$ (i.e., minimum with respect to $\alpha_1$).
The functions $\alpha_p$, for $p>0$, and $\alpha_{\infty}$ are 
$\ell_1$-optimizable.
So far in this paper we have been using the $\alpha_1$ cost function
to determine the weight of the MStT.
\end{definition}

\begin{definition}[Brazil \etal $2015$ \cite{brazil2015generalised}]
\label{def: min tree wrt alpha}
Consider a set $P$ of $n$ points in $\mathbb{R}^2$, a norm $||\cdot||$, and
a symmetric $\ell_1$-optimizable function $\alpha$. 
Let $S$ be a set of at most $k$ points in $\mathbb{R}^2$, and let $\mathbb{T}$ be
a spanning tree on $P \cup S$ with topology $\tau$.
Let $S$ and $\tau$ be such that 
$\alpha(\textbf{e}_{\tau,P,S})=\min_{\tau', S'}\alpha(\textbf{e}_{\tau',P,S'})$.
We say such a tree $\mathbb{T}$
is a \emph{minimum $k$-Steiner tree on $P$ \textbf{with respect
to $\alpha$}}.
\end{definition}

\begin{corollary}[Brazil \etal $2015$ \cite{brazil2015generalised}, Corollary $2$]
\label{cor: mst also gen k st tree}
Given a set $P$ of $n$ points in $\mathbb{R}^2$, a norm $||\cdot||$, and
a symmetric $\ell_1$-optimizable function $\alpha$, 
let $\mathbb{T}$ be a minimum $k$-Steiner tree on $P$ with respect
to $\alpha$ and
let $S \subset \mathbb{R}^2$ be the set of at most $k$ Steiner points from $\mathbb{T}$
(i.e., for the vertex set $V(\mathbb{T})$ of $\mathbb{T}$, $S = V(\mathbb{T}) \setminus P$).
Every MST on $P \cup S$ is a minimum $k$-Steiner tree on $P$ with respect
to $\alpha$.
\end{corollary}
It follows from \cref{theorem: main result} that \cref{alg: algo}
computes a tree with topology $\tau^{\spadesuit}$ and set 
of Steiner points $S^{\spadesuit}$ of at most one element such that 
$\alpha_1(\textbf{e}_{\tau^{\spadesuit},P,S^{\spadesuit}}) = \min_{\tau,S}\alpha_1(\textbf{e}_{\tau,P,S})$.
From the definition of $\ell_1$-optimizable, it follows that 
$\alpha(\textbf{e}_{\tau^{\spadesuit},P,S^{\spadesuit}}) = \min_{\tau, S}\alpha(\textbf{e}_{\tau,P,S})$.
By \cref{cor: mst also gen k st tree}, 
it follows that our MStT is also the solution that minimizes the weight
of the solution with respect to cost function $\alpha$.
A similar argument holds for
the algorithm of \cref{sec: k st points} and
\cref{corollary: j norm curves}.
This discussion leads to \cref{corollary: other costs}.

\generalCorollary

\clearpage

\subsection{Curving \texorpdfstring{$\gamma$}{the Input gamma}}
\label{sec: from lines to curves}

Examining \cref{alg: algo}, we note that there are two parts involving 
$\gamma$: computing labelled intervals along $\gamma$,
and computing an optimal Steiner point for each interval.
By examining the properties of lines that were being exploited to perform
certain operations in $O(1)$ time, we realize our results are 
applicable to a more general set of input constraints than
lines.
We now let $\gamma$ be more general than a line.
Abusing notation, we will refer to $\gamma$ as a \emph{curve}, 
even if it has endpoints 
(in which case it may be closer to an \emph{arc}).
We keep our general position assumption that though $\gamma$ may intersect
the rays and segments of the OVDs multiple times, each intersection is a 
single point.

\subsubsection*{Computing Labelled Intervals on \texorpdfstring{$\gamma$}{the Input gamma}}

\begin{description}
	\item[\textbf{Compute Intersections of $\gamma$ with Lines/Rays/Segments}]
 ~\\
The first property we used was the zone theorem 
which said that our input line
intersects an arrangement of $O(n)$ lines (i.e., the OVD) $O(n)$ times.
Let $O(g)$ be the number of 
times $\gamma$ intersects an arrangement
of $O(n)$ lines.
Since we are intersecting $\gamma$ with six OVDs, the number of intervals
created on $\gamma$ is $O(g)$.
We also used the fact that each intersection can be computed in $O(1)$
time and space.
Let $O(h)$ be the time and $O(h_{sp})$ be the space
it takes to compute an intersection between $\gamma$ and a line.
Then the time to compute all of the intersection points on $\gamma$ is 
$O(gh)$ and the space is $O(g + h_{sp})$.

\begin{itemize}
	\item $O(g)$: zone of $\gamma$ in an arrangement of $O(n)$ lines
	\item $O(h)$: time to compute the intersection of $\gamma$ and a line
	\item $O(h_{sp})$: space to compute the intersection of $\gamma$ and a line
	\item $O(gh)$: time to compute all of the intersection points on $\gamma$
	\item $O(g+h_{sp})$: space to compute all of the intersection points on $\gamma$ 
\end{itemize}

We use $O(\beta)$ to denote the time it takes to compute the intersection
points of $\gamma$ with the OVDs (i.e., $\beta = gh$), and we use 
$O(\beta_{sp})$ to denote the space.

	\item[\textbf{Create Intervals}]
 ~\\
	Those computed intersection points are actually the endpoints of 
	intervals on $\gamma$.
Next, we sorted the interval endpoints to construct the intervals using
the fact that the line was $x$-monotone.
Let $m$ be the time and $m_{sp}$ be the space it takes
to break $\gamma$ up into $c$ 
univariate polynomial functions.
For example, we can break a circle into two $x$-monotone components, each
starting at the left-most point and going to the right-most point, but
one moving clockwise and the other counterclockwise.
Overall, all of these components contain $O(g + c)$ 
points to be sorted by $x$-coordinate,
so we spend $O((g+c)\log (g+c))$ time and 
$O(g + c)$ space to sort the interval endpoints. 

Once we have sorted the points, we need to
figure out in which of our $c$ $x$-monotone components
each point belongs.
Let $w$ be the most time and $w_{sp}$ be the most space
required to evaluate one of our $c$ functions at a given $x$-coordinate.
Then we create the $O(g + c)$
intervals in $O(c(g + c)w)$ time and $O(g + c + w_{sp})$ space
as follows.
For each of our $c$ components, we walk along our sorted list of points, 
evaluate the current component's function at the $x$-coordinate of the
next point in the sequence, and compare the result against the $y$-coordinate
of the tested point.
If it matches, then this point will be an interval endpoint on this
component of $\gamma$.
In this way, we separate
the interval endpoints into the components in which they exist on $\gamma$.
Then we walk over the resulting lists to match interval endpoints 
(i.e., decide where intervals start and end).

\begin{itemize}
	\item $c$: number of univariate polynomial functions into which $\gamma$ is decomposed
	\item $O(m)$: the time to decompose $\gamma$ into $c$ univariate polynomial functions
	\item $O(m_{sp})$: the space to decompose $\gamma$ into $c$ univariate polynomial functions
	\item $O(g + c)$: the number of intervals into which $\gamma$ is decomposed
	\item $O((g+c)\log (g+c))$: the time to sort the interval endpoints 
	\item $O(g + c)$: the space to sort the interval endpoints
	\item $O(w)$: the most time to evaluate one of our $c$ functions at a given $x$-coordinate
	\item $O(w_{sp})$: the most space to evaluate one of our $c$ functions at a given $x$-coordinate
	\item $O(c(g + c)w)$: the time to create the intervals
	\item $O(g + c + w_{sp})$: the space to create the intervals
\end{itemize}
		
\end{description}

We use $O(\lambda)$ to denote the time it takes to compute the $O(g+c)$ 
intervals on $\gamma$ 
and we use $O(\lambda_{sp})$ to denote the space
(i.e., $O(m + (g+c)\log (g+c) + c(g + c)w)$ time and $O(m_{sp} + g + c + w_{sp})$ space).

\subsubsection*{Computing Optimal Steiner Points in the Intervals}

Lastly, we took advantage of the fact that in
$O(1)$ time and space we could find our candidate points in each
interval by computing the zeroes of the derivative of the distance
function $d_{\mathbb{P}}(x)$ that sums the distances from the points 
associated with the interval to the point $x$ on the line.

As before, consider an interval and its set of 
potential neighbours $P' \subset P$ of
size at most six.
For each subset $\mathbb{P}$ of $P'$ of size three to six (of which
there are $O(1)$ such subsets),
we compute 
a candidate optimal Steiner point
in the interval
on $\gamma$.
Although we have $c$ functions representing $\gamma$, 
here we abuse notation and refer to them all as $\gamma$.
Recall that $\gamma$ is a function whose parameter corresponds to the
$1$-dimensional position along $\gamma$.
The formula for the distance function will depend on the norm.
For example, with the Euclidean norm,
the sum-of-distances  function  looks like it did before:
$d_{\mathbb{P}}(x) = \sum_{a \in \mathbb{P}}\sqrt{(a_x - x)^2 + (a_y - \gamma (x))^2}$.
Over all intervals on $\gamma$, let $t$ be the most time and 
$t_{sp}$ be the most space used to find a candidate Steiner point
for a given interval.
This involves finding the zeroes of the derivative of $d_{\mathbb{P}}(x)$
and thus depends on the degree of the function $\gamma(x)$ and the norm being used.
We can get a simple upper bound on the amount of time and space
used to compute the candidate optimal Steiner points 
and the sums of the edge-lengths 
for edges incident to those Steiner points 
in the MStT containing those candidates:
$O(t(g + c))$ time and $O(t_{sp} + g + c)$ space.
Let $q$ be the most time and $q_{sp}$ be the most space used to 
perform step 
\ref{k st solve fixed topology steiner problem} 
when restricting the Steiner points to
lie on $\gamma$ (i.e., 
solve the fixed topology $k$-Steiner tree problem).
Note that step \ref{k st update MST info} 
(i.e., computing the minimum $F$-fixed spanning tree)
is performed using operations on weighted graphs 
and so is not affected by changing $\gamma$ to a curve.

\begin{restatable}{corollary}{GeneralCorollaryTwo}
\label{corollary: general}
Given:
\begin{itemize}
    \item a set $P$ of $n$ points in $\mathbb{R}^2$;
    \item a norm $\Vert \cdot \Vert$
that is compliant to Restrictions \ref{rest:one}, \ref{rest:two}, 
and \ref{rest:three};
    \item a symmetric $\ell_1$-optimizable cost function $\alpha$;
    \item a curve $\gamma$ with space complexity $\gamma_{sp}$
\end{itemize}
\cref{alg: algo} computes 
in $O(n\log n + \beta + \lambda + t(g + c))$ time 
and $O(n + \gamma_{sp} + \beta_{sp} + \lambda_{sp} + t_{sp})$ 
space 
a minimum-weight tree
with respect to $\alpha$ and $\Vert \cdot \Vert$ that  connects all points in $P$ 
using at most one extra point $s \in \gamma$.

For a constant integer $k>1$, 
the algorithm of \cref{sec: k st points} solves
the restricted $k$-Steiner tree problem in
$O((g+c)^{k}q + \beta + \lambda + n\log n)$ time 
and $O(n + \gamma_{sp} + \beta_{sp} + \lambda_{sp} + q_{sp})$ space
with a Steiner set $S$ of at most $k$
points from $\gamma$.
\end{restatable}

This means, for example, that when $\gamma$ is a curve such as a conic
(e.g., an ellipse, or a hyperbola), or a constant-degree polynomial function, 
or a curve that can be split into $O(1)$ constant-degree polynomial
functions and $O(1)$ loops/self-intersections that can be decomposed in $O(1)$ time and space into $O(1)$
univariate functions described by polynomials of constant degree 
(such as the \textit{Folium of Descartes}),
\cref{alg: algo} still runs in $O(n\log n)$ time\footnote{Once again, we assume that, 
using the univariate polynomial functions into which we have decomposed 
$\gamma(x)$ in the equations outlined in the previous sections, 
the derivatives and roots required can be computed in constant time and space.}
and $O(n)$ space
with a norm $\Vert \cdot \Vert$
that is compliant to Restrictions \ref{rest:one}, \ref{rest:two}, 
and \ref{rest:three}, such as the Euclidean norm,
and an $\ell_1$-optimizable cost function $\alpha$.

We can also extend our result to the problem in which we are given a set of 
$j$ curves $\Gamma = \{\gamma_{1}, \ldots \gamma_{j}\}$.
Let the runtime for computing the intervals on $\gamma_i \in \Gamma$ 
be $O(\beta_i + \lambda_i)$ and the space be $O(\beta_{sp,i} + \lambda_{sp,i})$, 
let $\mu$ be $\max_i(\beta_i + \lambda_i)$, and let $\mu_{sp}$
be $\max_i(\beta_{sp,i} + \lambda_{sp,i})$.
Let $(g+c)$ be the maximum number of intervals into which any $\gamma_i$ is decomposed.
Let $t$ be the most time and $t_{sp}$ be the most space used 
to find a candidate Steiner point over all intervals
of all $\gamma_i \in \Gamma$.
Let $q$ be the most time and $q_{sp}$ be the most space used to 
solve the fixed topology $k$-Steiner tree problem
when restricting the Steiner points to
lie on curves in  $\Gamma$.

\finalCorollary

\section{Conclusion}
\label{sec: conclusion}
We showed that given a set $P$ of $n$ points in $\mathbb{R}^2$ 
and a line $\gamma$ in $\mathbb{R}^2$,
\cref{alg: algo} computes in optimal $\Theta(n\log n)$ time and 
optimal $\Theta(n)$ space a minimum-weight tree connecting 
all points in $P$ using at most one  extra point  $s \in \gamma$.
We noted that this result can extend to a set
$\Gamma = \{\gamma_{1}, \ldots \gamma_{j}\}$ of $j$ lines,
and that by running \cref{alg: algo} for each $\gamma \in \Gamma$, 
in $O(jn\log n)$ time and $O(n + j)$ space 
a minimum-weight tree is computed that connects all points in $P$ 
using at most one  extra point 
$s \in \bigcup_{i=1}^{j} \gamma_i$.

Next we reviewed the algorithm to solve the $k$-Steiner tree problem
presented by Brazil et al.\ \cite{brazil2015generalised} that 
runs in $O(n^{2k})$ time and $O(n^2)$ space for $k=1$,
and $O(n^3)$ space for $k>1$.
We showed that we can replace one of their $O(n^2)$ 
time and space steps and one of their $O(n^3)$ time and space
steps with preprocessing steps
that require $O(n\log n)$ time and $O(n)$ space.
This replacement does not change the asymptotic time and
space complexities for the $1$-Steiner tree algorithms
of Georgakopoulos and Papadimitriou \cite{oneStTreeProb} 
or Brazil \etal\ \cite{brazil2015generalised},
nor the time complexity of the $k$-Steiner tree algorithm
of Brazil \etal\ \cite{brazil2015generalised} for $k>1$, but it lowers
the space complexity from $O(n^3)$ to $O(n^2)$ for $k>1$.
We then showed how to adapt the $k$-Steiner tree
algorithm to our setting (i.e., where the Steiner points can only
be chosen from a set of $j$ input lines), allowing us to solve our restricted
$k$-Steiner tree problem in $O((jn)^k)$ time and $O(jn)$ space.

We then pointed out that the more general results of 
Brazil \etal\ \cite{brazil2015generalised} also apply to our scenario,
thereby extending our results to the same norms and cost functions 
as the algorithm of Brazil \etal, abiding by the restrictions
they laid out.

Lastly, we showed that our results apply when the set of input restrictions is
not a set of lines, but a set of $j$ curves.
In this case, the running time and space bounds come to depend on $j$, the 
runtime and space of certain primitive operations, and the complexity of the
zone of the curves in an arrangement of lines.

It is an open question whether or not the unrestricted 
$1$-Steiner tree problem in $\mathbb{R}^2$
can be solved in $o(n^2)$ time by 
using multiple constraint lines and the results
of \cref{alg: algo} to guide a search.

\section*{Acknowledgements}
The authors thank Jean-Lou De Carufel for helpful discussions.

\clearpage

\bibliographystyle{plain}
\bibliography{restricted-steiner}

\end{document}